\documentclass[onecolumn, 11pt, draftclsnofoot]{IEEEtran}
\usepackage{graphicx}
\usepackage{amssymb}
\usepackage{amsfonts}
\usepackage{amsmath}

\usepackage{color}

\usepackage{multirow}
\usepackage{setspace}
\usepackage{psfrag}
\usepackage{subcaption}

\usepackage{booktabs}

\renewcommand{\baselinestretch}{1.65}

\usepackage[ruled,vlined]{algorithm2e}
\usepackage{mathrsfs}

\definecolor{light-gray}{gray}{0.90}

\newtheorem{Proposition}{Proposition}

    \def\Complex{{\rm\rule[.23ex]{.03em}{1.1ex}\kern-.3em{C}}}

    \newcommand{\be}{\begin{equation}} \newcommand{\ee}{\end{equation}}
    \newcommand{\bea}{\begin{eqnarray}} \newcommand{\eea}{\end{eqnarray}}
    \newcommand{\benum}{\begin{enumerate}} \newcommand{\eenum}{\end{enumerate}}

    \newcommand{\qa}{{\bf a}}
    \newcommand{\qb}{{\bf b}}

    \newcommand{\qg}{{\bf g}}
    \newcommand{\qh}{{\bf h}}

    \newcommand{\qv}{{\bf v}}

    \newcommand{\qy}{{\bf y}}
    \newcommand{\qz}{{\bf z}}

    \newcommand{\qF}{{\bf F}}
    
    \newcommand{\qH}{{\bf H}}

    \newcommand{\qP}{{\bf P}}
    \newcommand{\qQ}{{\bf Q}}

    \newcommand{\qpsi}{{\boldsymbol \psi}}

    \newcommand{\qtheta}{{\boldsymbol \theta}}
    \newcommand{\qTheta}{{\boldsymbol \Theta}}

    \newcommand{\qGamma}{{\boldsymbol \Gamma}}

    \newcommand{\qtau}{{\boldsymbol \tau}}

        \newcommand*{\argmax}{\operatornamewithlimits{argmax}\limits}

\begin{document}

\title{Fast Antenna and Beam Switching Method for mmWave Handsets with Hand Blockage}

\author{~Wan-Ting~Shih,~Chao-Kai~Wen,~\IEEEmembership{Member,~IEEE},~Shang-Ho~(Lawrence)~Tsai,~\IEEEmembership{Senior Member,~IEEE}, and~Shi~Jin,~\IEEEmembership{Senior Member,~IEEE}

    \thanks{{W.-T.~Shih} is with the Institute of Electrical Control Engineering, National Chiao Tung University, Hsinchu 30010, Taiwan, Email: {\rm  sydney2317076@gmail.com}.}
    \thanks{{C.-K.~Wen} is with the Institute of Communications Engineering, National Sun Yat-sen University, Kaohsiung 80424, Taiwan, Email: {\rm chaokai.wen@mail.nsysu.edu.tw}.}
    \thanks{{S.-H.~Tsai} is with the department of Electrical Engineering, National Chiao Tung University, Hsinchu 30010, Taiwan,
        Email: {\rm shanghot@mail.nctu.edu.tw}.}
    \thanks{{S.~Jin} is with the National Mobile Communications Research Laboratory, Southeast University, Nanjing 210096, P. R. China, Email: {\rm  jinshi@seu.edu.cn}.}
}

\maketitle
\vspace{-1.5cm}
\begin{abstract}
\footnotesize
Many operators have been bullish on the role of millimeter-wave (mmWave) communications in fifth-generation (5G) mobile broadband because of its capability of delivering extreme data speeds and capacity. However, mmWave comes with challenges related to significantly high path loss and susceptibility to blockage.
Particularly, when mmWave communication is applied to a mobile terminal device, communication can be frequently broken because of rampant hand blockage. Although a number of mobile phone companies have suggested configuring multiple sets of antenna modules at different locations on a mobile phone to circumvent this problem, identifying an optimal antenna module and a beam pair by simultaneously opening multiple sets of antenna modules causes the problem of excessive power consumption and device costs. In this study, a fast antenna and beam switching method termed Fast-ABS is proposed. In this method, only one antenna module is used for the reception to predict the best beam of other antenna modules. As such, unmasked antenna modules and their corresponding beam pairs can be rapidly selected for switching to avoid the problem of poor quality or disconnection of communications caused by hand blockage. Thorough analysis and extensive simulations, which include the derivation of relevant Cram\'{e}r-Rao lower bounds, show that the performance of Fast-ABS is close to that of an oracle solution that can instantaneously identify the best beam of other antenna modules even in complex multipath scenarios. Furthermore, Fast-ABS is implemented on a software defined radio and integrated into a 5G New Radio physical layer. Over-the-air experiments reveal that Fast-ABS can achieve efficient and seamless connectivity despite hand blockage.
\end{abstract}


\section{Introduction}

Current research in wireless cellular networks is developing technologies that can meet the global surge in demand for mobile data \cite{forecast2019cisco}. In comparison with congested bands below 6\,GHz, millimeter-wave (mmWave) communications \cite{mmWave1,mmWave2,mmWave3,mmWave4} provide more available bandwidths than existing 4G systems do; as such, mmWave communications have become a promising solution. They have been recommended as a key technology for 5G mobile broadband not only by recent research advances in new mmWave systems \cite{newmmWave1,newmmWave2,newmmWave3,newmmWave4,newmmWave5,newmmWave6,newmmWave7} but also by popular network operators and vendors \cite{vendors} who had done large-scale field trials \cite{mmWave1,Verizon}. Unfortunately, mmWave signals suffer from severe path losses because of high carrier frequencies. mmWave devices must relay via array antennas to focus their radio frequency (RF) energy through narrow directional beams and compensate for attenuation loss. In the case of narrow beams, small changes in body position relative to mobile devices can cause dramatic changes in signal strength. Connectivity can be lost frequently. Therefore, efficiently aligning the beam between a base station (BS) and user equipment (UE) is a critical task to realize mmWave communications in 5G systems \cite{newmmWave3,newmmWave4,newmmWave5,newmmWave6,newmmWave7}.

In 3GPP 5G New Radio (NR), a series of operations called beam management is standardized to establish and retain a suitable \emph{beam pair} that refers to a good connectivity between a transmitter beam direction and a corresponding receiver beam direction. Specifically, beam management includes three different levels: (i) in initial access \cite{IA1,IA2,IA3}, a preliminary beam pair should be established between a BS and a UE in downlink and uplink transmission directions before data are transmitted; ii) in beam adjustment \cite{BT1,BT2}, once a preliminary beam pair is established, the UE should continuously update the beam pair to maintain a good connection every time because of its rotation or slight movement; and (iii) in beam recovery, if a connection with the BS is broken due to the excessively violent movement of the UE, a recovery procedure should be performed to restore communication between the BS and the UE. Additional procedures targeting beam failure should be introduced by using the UE, and multiple beam scanning should be employed to complete beam recovery, which may require hundreds of milliseconds of waiting time to identify a new beam pair \cite{newmmWave7}. Therefore, avoiding a beam failure event is equivalent to reducing latency caused by beam recovery.

Beam alignment technology has shown remarkable advancements \cite{5GNR,Balevi-GLOBECOM19}, but its designs have been largely developed in free space propagation without a hand blockage effect. However, when mobile phones are used, users' fingers shield the mmWave array antenna, thereby failing to receive signals. To overcome this problem, a number of mobile phone companies \cite{Ouyang-Patent19,Qualcomm-17,Sony-18} have suggested configuring multiple sets of mmWave antenna modules on a mobile phone. As shown in Figure \ref{fig:handblockage}, unmasked antenna modules and their corresponding beam pairs can be selected for switching to avoid the problem of poor quality or disconnection of mmWave communication due to the shielding of the used antenna module. However, identifying an optimal antenna module and a beam pair by simultaneously opening multiple sets of antenna modules causes excessive power consumption and device costs. In general, one RF chain is a reasonable architecture for the initial generation of mmWave mobile phones \cite{turn_1}. Such an architecture can support the switching mechanism over antenna modules and receiver beams to select the one that maximizes antenna gain. Therefore, developing an effective switching antenna mechanism that can maintain the communication performance of mobile phones with low latency becomes a critical task.

Studies \cite{turn_1,turn_2} have suggested switching to another antenna module in turn and rescanning the best beam direction, i.e., reprocessing the initial access, for reception when an object blocks the received antenna module. In this approach, a long latency period can be required to find the best antenna module and build beam alignment. In another study \cite{op_1}, the knowledge of a user's handgrip of a mobile phone (e.g., when browsing or taking photos) is applied to switch antennas and design corresponding codebooks for reducing the number of searches. However, this method requires knowing the usage mode information of mobile phones, and performance still depends on the size of a search.

\begin{figure}
    \begin{center}
        \resizebox{3.6in}{!}{%
            \includegraphics*{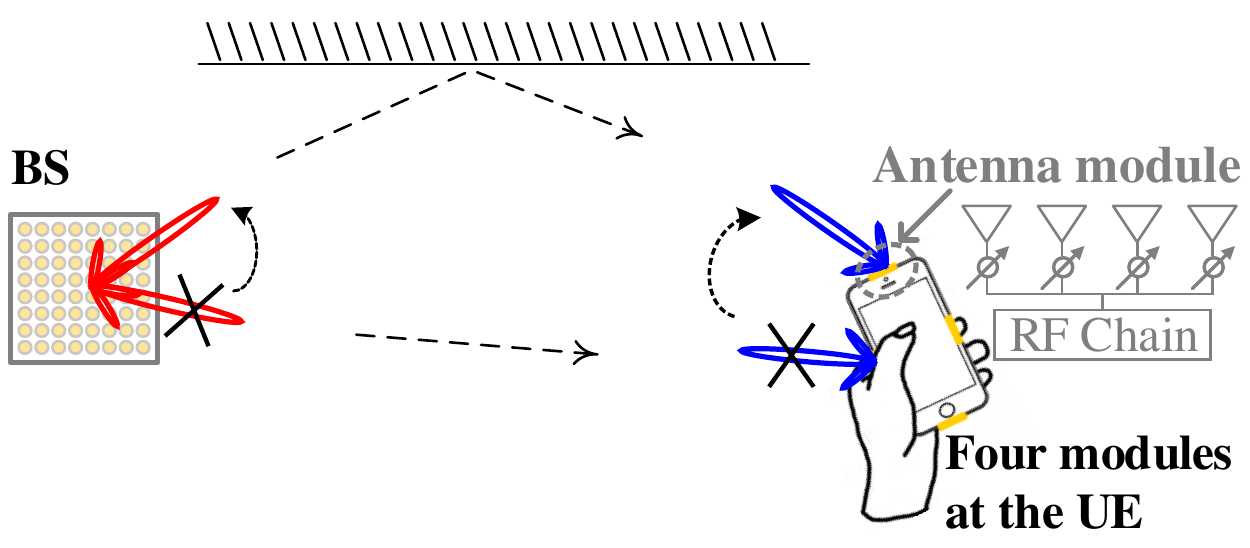} }%
        \caption{Avoiding hand blockage by switching an antenna module and a beam pair. }\label{fig:handblockage}
    \end{center}
\end{figure}

In this study, a {\bf fast} {\bf a}ntenna and {\bf b}eam {\bf s}witching method referred to as Fast-ABS is proposed for mmWave handsets with multiple antenna modules. Our method is basically built on the basis of two observations.
\begin{itemize}
    \item First, although the steering vector changes with different antenna modules and orientations, the underlying physical paths of signals propagating from the BS to the compact UE remain unchanged in free space. Therefore, if one antenna module can extract relevant parameters (i.e., complex gain and directionality) of paths, then the spatial information of a channel of other antenna modules can be reconstructed without scanning all antenna modules.

    \item Second, although the underlying propagation paths for a blocked antenna module can be complex because of the coupling and reflection caused by a user's fingers, their detailed properties are irrelevant because using a blocked antenna module should be avoided. Only a simple power detector is required to detect whether an antenna module is experiencing hand blockage.
\end{itemize}
With Fast-ABS, the properties of paths from one receiver antenna module, along with the power information of other antenna modules, may be extracted.
As a result, all the receiver beam directions of all antenna modules can be ranked in the UE through the channel measurement of a current antenna module.\footnote{
We have preliminary verified this idea in the conference version of this study \cite{Shih-ICC20}.}
If the UE can wait a while to probe a channel from different transmit beams of the BS, then the UE can have large candidates to obtain a better beam pair not only for antenna modules but also for transmitter beams.

To achieve the goal of the first observation, we infer the relevant parameters, i.e., complex gains and angle of arrival (AoA), of dominant paths by utilizing beam-specific channel responses measured under different beams. We derive the Cram\'{e}r-Rao lower bound (CRLB) of the mean square error (MSE) of relevant parameters estimated by beam-specific measurements. The CRLB shows that a low MSE can be achieved by selecting a proper set of beams so that relevant parameters can be obtained by utilizing a small number of received beams. Specifically, our results show that only three to four probings are required to infer relevant parameters. Furthermore, we derive the CRLB of the channel MSE reconstructed on the basis of the extracted parameters, and analysis is carried out by simulating the characteristics of mmWave wireless channels. Our results demonstrate that the receiving performance achieved by Fast-ABS is almost the same as that observed by exhaustively scanning (ES) through all angles and antenna modules. Finally, we implement and evaluate Fast-ABS in a 5G NR device that supports two antenna modules in various static and dynamic settings. Our experiments show that the signal-to-noise ratio (SNR) under Fast-ABS is close to that under an optimal beam obtained by ES.

\section{System Model and Problem Formulation}

We consider a 3GPP 5G NR-compatible mmWave communication system with a BS and a UE. A synchronization signal block (SSB) is periodically broadcasted via the BS by using $S$ (directional) beams. SSB signals are acquired via the UE to establish a communication link by utilizing $M$ (directional) beams during the initial access. The UE is equipped with $P$ antenna modules, and a finite-sized analog beam codebook is used for beamforming. We assume that the UE can focus on one beam (or direction) at a time because of component cost and power consumption considerations, i.e., one antenna module and one beam are employed at a time. If a signal is transmitted from the BS via the $s$-th beam over a channel with $L_s$ paths, then
the channel observed at the UE through the $m$-th receiver beam of the $p$-th antenna module is given by
\begin{equation} \label{eq:hm}
h_{s,m}^{[p]}(t) = \sum_{l=1}^{L_s} g_{s,m,l}^{[p]} \cdot A_{m}^{[p]}{\left(\theta_{s,m,l}^{[p]}\right)} \cdot \delta{\left(t-\tau_{s,m,l}^{[p]}\right)},
\end{equation}
where $g_{s,m,l}^{[p]}$, $\tau_{s,m,l}^{[p]}$, and $\theta_{s,m,l}^{[p]}$ denote the corresponding complex-valued channel gain, AoA, and time of arrival (ToA) of the $l$-th path, respectively; and $A_{m}^{[p]}(\theta)$ represents the receiver beamforming gain along the $\theta$ direction. In practice, the signal path has azimuth and elevation components in 3D space, but to make the representation simple, we omit the elevation direction.
Our proposed method can be extend to the case with elevation components. Channel state information (CSI) is obtained with the UE by using different receiver beams and antenna modules during initial access.

Once the initial beam pair is established through the initial access, an appropriate beam pair is maintained for data communication through a process called beam adjustment. Notably, this process is accomplished by measuring a downlink CSI reference signal (CSI-RS) and does not interrupt transmission. Transmitter beams entrain the information of the transmission beam direction in their respective configured CSI-RSs. After transmitter- and receiver- side beams are swept, performance can be obtained from each beam pair via a receiver. Afterward, the new beam pair can be confirmed by the result of CSI-RS.

The aforementioned measurements can be performed in the frequency domain because a 5G NR system is operated via orthogonal frequency-division multiplexing (OFDM). The channel response of \eqref{eq:hm} in the frequency domain is given by
\begin{equation} \label{eq:Hm}
H_{s,m}^{[p]}[f_k] = \sum_{l=1}^{L_s} g_{s,m,l}^{[p]} \cdot A_{m}^{[p]}{\left(\theta_{s,m,l}^{[p]}\right)} \cdot e^{-j 2 \pi \tau_{s,m,l}^{[p]} f_k },
\end{equation}
where $f_k$ is the $k$-th sub-carrier frequency. We can estimate $ H_{s,m}^{[p]}[\cdot] $ at subcarriers $ f_1, f_2, \cdots, f_{N_s} $ by dividing subcarrier outputs
by the known RS. Over a set of receiver beams ${\{m = 1, 2, \cdots, M\}}$, the receiver can obtain an ${M \times N_s}$ CSI matrix of the $p$-th antenna module as follows:
\begin{equation} \label{eq:bH}
\qH_{s}^{[p]} = \left[
\begin{matrix}
H_{s,1}^{[p]}[f_1] &  H_{s,1}^{[p]}[f_2] & \cdots & H_{s,1}^{[p]}[f_{N_s}] \\
\vdots &    & \ddots   \\
H_{s,M}^{[p]}[f_1] &  H_{s,M}^{[p]}[f_2] & \cdots & H_{s,M}^{[p]}[f_{N_s}] \\
\end{matrix}
\right].
\end{equation}
By integrating \eqref{eq:Hm}, \eqref{eq:bH} can be expressed as
\begin{equation} \label{eq:H_mat}
\qH_{s}^{[p]} = \sum_{l=1}^{L_s} g_{s,m,l}^{[p]} \qa^{[p]}(\theta_{s,m,l}^{[p]}) \qb^{H}(\tau_{s,m,l}^{[p]}),
\end{equation}
where
\begin{equation}
\qa^{[p]}(\theta) \triangleq \left[  A_{1}^{[p]}(\theta)\, A_{2}^{[p]}(\theta)\, \cdots \, A_{M}^{[p]}(\theta) \right]^T, ~~
\qb(\tau) \triangleq \left[   e^{-j 2 \pi \tau f_1 } \, e^{-j 2 \pi \tau f_2 } \, \cdots \, e^{-j 2 \pi \tau f_{N_s} } \right]^{H}.
\end{equation}

This study mainly focuses on the period of beam adjustment rather than the period of initial access because the former is critical for seamless connectivity under hand-blockage. As illustrated in Figure \ref{fig:handblockage}, the UE as the receiver has four antenna modules placed over four edges. Currently available mobile devices are restricted to powering only one antenna module at a time, so simultaneously and immediately determining the antenna module and beam pair is a challenging task. Two major problems have to be addressed: First, if an object blocks the received antenna module, then the UE switches to another module and re-scan the best receiver beam. Scanning may be invoked persistently among multiple antenna modules over $M \times P$ beams, exacerbating latency. Although the use of a large number of antenna modules can theoretically provide good coverage, if module switching cannot be performed with a low beam management overhead, then the function becomes difficult, becoming a detriment rather than a benefit. Second, even worse, the optimal beam pair between the BS and the UE does not necessarily correspond to transmitter and receiver beams that are \emph{physically} directed at each other. A direct path may be blocked, and a reflective path may provide better connectivity to the other antenna module (Figure \ref{fig:handblockage}). Re-scanning likely invokes huge beam pairs between the BS and multiple antenna modules (i.e., $S \times M \times P$ beams), resulting in unacceptable latency.

\section{Fast-ABS}
\begin{figure*}
    \begin{center}
        \resizebox{6.75in}{!}{%
            \includegraphics*{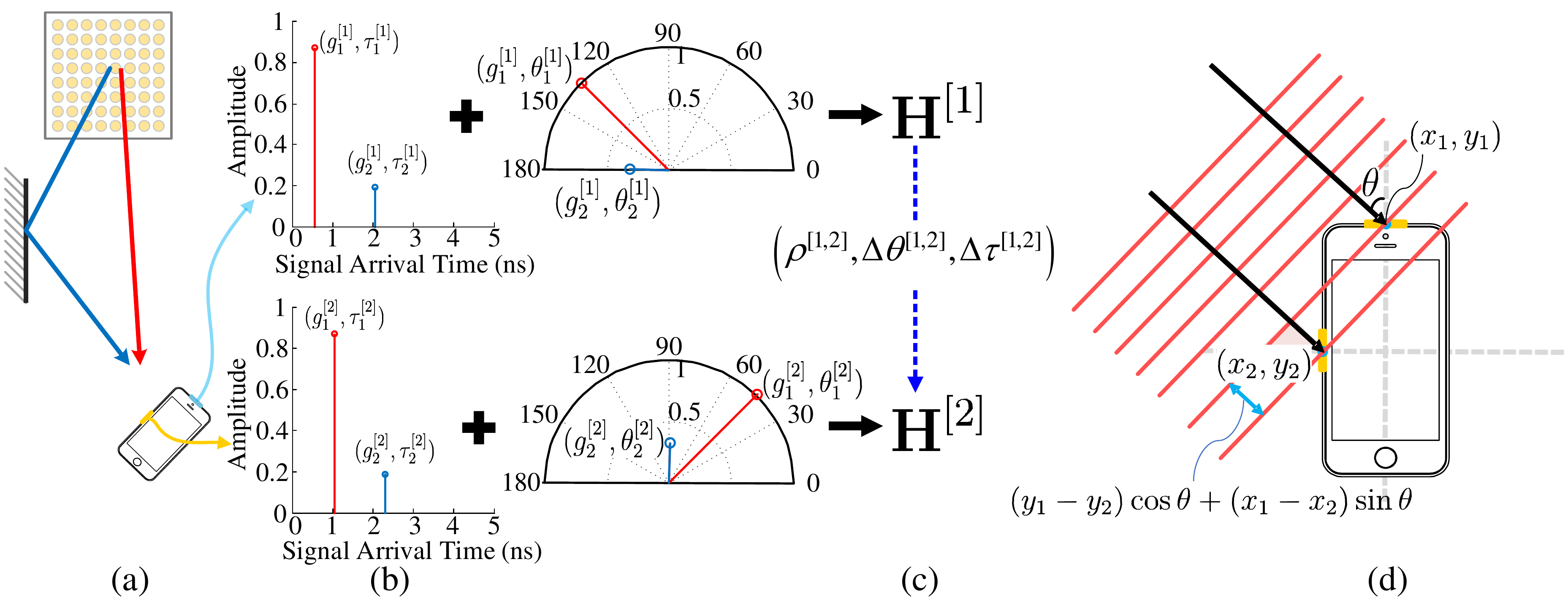} }%
        \caption{Rationale behind Fast-ABS: (a) Two antenna modules share the same physical paths to reach a mobile device.
            (b) Measured channel responses show the ToA and AoA of the two paths under both antenna modules.
            (c) Although the two channels look extremely different, we can map the three-tuples of the paths back to the channels of antenna module 1 to predict the channels of antenna module 2.
            (d) Diagram of AoA and ToA between two different antenna modules. }\label{fig:Intuition}
    \end{center}
\end{figure*}

\subsection{Rationale}

To understand the rationale of Fast-ABS, we first analyze the CSI represented in \eqref{eq:H_mat}, i.e., the channel of a single transmission beam reaching an antenna module. We define each antenna module consisting of an $N$ antenna element uniform linear array\footnote{Relative to the configuration with a single antenna on each antenna module, the usage of an array antenna achieves a directional beam easily and does not reduce the beamforming gain due to hand blockage.}
with a uniform separation $d$. Each antenna element is equipped with a phase shifter to achieve beam steering capabilities, as shown in Figure \ref{fig:handblockage}.
Then the beam pattern given the analog beam codebook ${\{w_m(n), n=0, \ldots, N-1 \}}$ of the $m$-th beam can be expressed as \cite{beamforming}:
\begin{equation}\label{eq:Am}
A_{m}^{[p]}(\theta) = \sum_{n=0}^{N-1} w_m(n) e^{-j 2 \pi n \frac{d \cos \theta }{\lambda}},
\end{equation}
where $\lambda$ is the wavelength of the wireless signal. Clearly, given a fixed array configuration and beam codebook, the beam pattern is irrelevant to the antenna module and only dependent on its AoA $\theta$. Hence, the superscript ${[p]}$ can be removed from $A_{m}^{[p]} $ and $\qa^{[p]} $ in the subsequent descriptions; i.e.,
\begin{equation}
 A_{m}^{[p]}(\cdot) \rightarrow A_{m}(\cdot) , ~ \qa^{[p]}(\cdot)  \rightarrow \qa(\cdot) .
\end{equation}
$\qa(\theta)$ and $\qb(\tau)$ in \eqref{eq:H_mat} is deterministic, so the $l$-th path can be fully characterized by three-tuples in the form $(g_{s,m,l}^{[p]}, \theta_{s,m,l}^{[p]}, \tau_{s,m,l}^{[p]})$. As a result, CSI, $\qH_{s}^{[p]}$, can be determined as long as three-tuples for each of the $L_s$ paths in the mmWave channel are available.

Although changing the receiver beams and antenna modules lead to different channel measurements at the receiver, the underlying propagation paths traversed by each receiver beam and antenna module should be invariant. Different receiver beams on the same antenna module have varying analog beam weights $\{ w_m(n) \}_{m=1}^{M}$ but have fixed channel gain, ToA, and AoAs. Different antenna modules in various locations also have varying but fixed orientations. Therefore, if a receiver can extract the three-tuples of each path, then the receiver can reconstruct an estimate of the channel for any receiver beams and antenna modules. With CSI matrices associated with all antenna modules, the receiver can determine the optimal antenna module and beam direction. The details and implementation of this observation are described in the succeeding subsections.

\subsection{Relationship of Three-tuples between Antenna Modules}

We describe how the three-tuples of each path are associated with antenna modules. Different receiver beams on the same antenna module have the same channel gain, ToA, and AoAs. Therefore, we can remove the receiver beam index $m$ from the three-tuple parameters:
\begin{equation}
 g_{s,m,l}^{[p]}\rightarrow g_{s,l}^{[p]}, ~\theta_{s,m,l}^{[p]}\rightarrow \theta_{s,l}^{[p]}, ~ \tau_{s,m,l}^{[p]} \rightarrow \tau_{s,l}^{[p]}.
\end{equation}
As long as $g_{s,l}^{[p]}$ is available, we can obtain the entire beam-patterns $\qa^{[p]}$ even without scanning all the receiver beams in the analog beam codebook. We discuss a method to estimate the three-tuple parameters later in Section III.D. We independently apply the principles described in this subsection to each transmitter beam independently. To avoid the key features of Fast-ABS being obfuscated by the unavoidably intricate notation, we focus on a transmitter beam and omit transmitter beam index $s$ from channel-related parameters; i.e., we simply use $(g_{l}^{[p]}, \theta_{l}^{[p]}, \tau_{l}^{[p]})$ in this subsection.

For example, we consider an environment where the BS transmits one beam to the UE traversing two propagation paths, Figure \ref{fig:Intuition}(a). If antenna modules 1 (blue) and 2 (yellow) at the receiver can simultaneously receive signals, then we let the corresponding three tuples of the signal paths at antenna module 1 be $(g_{1}^{[1]}, \theta_{1}^{[1]}, \tau_{1}^{[1]})$ and $(g_{2}^{[1]}, \theta_{2}^{[1]}, \tau_{2}^{[1]})$, and the antenna module 2 be $(g_{1}^{[2]}, \theta_{1}^{[2]}, \tau_{1}^{[2]})$ and $(g_{2}^{[2]}, \theta_{2}^{[2]}, \tau_{2}^{[2]})$, as shown in Figure \ref{fig:Intuition}(b).
The three quick observations can be found:
\begin{itemize}
    \item (OB-1). The AoA of each path is rotated at a fixed angle between antenna modules 1 and 2. Specifically, we have $\theta_{1}^{[2]} = \theta_{1}^{[1]} + 90^{\circ}$ and $\theta_{2}^{[2]} = \theta_{2}^{[1]} + 90^{\circ}$.

    \item (OB-2). A fixed delay shift seems absent in the ToA between antenna modules 1 and 2; i.e., $\tau_{1}^{[2]} - \tau_{1}^{[1]} \neq \tau_{2}^{[2]} - \tau_{2}^{[1]}$.

    \item (OB-3). The channel gain of each path between antenna modules 1 and 2 is identical, i.e., $g_{1}^{[1]} =  g_{1}^{[2]} $ and $g_{2}^{[1]} = g_{2}^{[2]}$.

\end{itemize}
These observations are described in detail as follows.

{\bf (OB-1)} The (absolute) AoAs of paths arriving at the UE are determined by the position of the UE and
its surrounding environment. Different AoAs between antenna modules result from the rotation angles of antenna modules on the UE. Fortunately, the position of antenna modules on a mobile device is fixed after hardware configuration is completed.
Therefore, we can predict or measure the fixed angular difference between antenna modules. For example, we consider that the two antenna modules in Figure \ref{fig:Intuition}(d) are placed on the edge of the UE. The AoA difference between antenna modules 1 and 2 is about $90^\circ$. We verify this observation through a practical measurement in Section V.B.

{\bf (OB-2)} Although ToAs between antenna modules 1 and 2 do not have a fixed delay shift, their time difference can be determined from the AoA and the position of the antenna module in accordance with the following proposition.
\begin{Proposition}
    Assume that the coordinates of antenna modules are central at $(x_1,y_1)$ and $(x_2,y_2)$, as shown in Figure \ref{fig:Intuition}(d).
    If the AoA of a path to antenna module 2 is $\theta$, then the propagation delay time
    between the two antennas module is
    \begin{equation}  \label{eq:delta_tau}
    \Delta\tau  = |y_1-y_2|\cos\theta-|x_1-x_2|\sin\theta.
    \end{equation}
\end{Proposition}
\begin{proof}
    The red line through $(x_1,y_1)$ cuts out a triangle with the acute angle $\theta$, the opposite side $|y_1-y_2|$, and the adjacent side $|y_1-y_2| \cot\theta $. Then, the hypotenuse of the (red) triangle is obtained as $ |y_1-y_2| \cot\theta - |x_1-x_2| $. Finally, the propagation delay can be calculated by $ {\sin\theta} (|y_1-y_2| \cot\theta - |x_1-x_2|)$, which can be rearranged as \eqref{eq:delta_tau}.
\end{proof}

{\bf (OB-3)} When AoAs and ToAs are determined via the above descriptions, complex channel gains do not need to absorb phase differences because of the position and rotations of antenna modules. Hence,
the complex channel gains of antenna modules should be identical under free-space propagation.
Even when the channel gains of antenna modules may differ from those of a linear scalar or phase offset, such differences are path independent and do not affect the best antenna module and beam selection. However, if hand blockage is considered, the above assumption is not valid anymore The loss from hand blockage on antenna gains can reach up to ${20-25}$ dB \cite{turn_1}, implying that an antenna module without hand blockage should be selected as the primary receiver antenna.
In this case, we should measure the channel gain of other antenna modules. This function can be simply achieved by using a power detector. This detector is applied at the output of a (image rejection) bandpass filter. Each antenna module is equipped with one power detector and used to measure power based on a pseudo-omni beam. As a result, the channel gain ratio between antenna modules can be obtained by utilizing the power detectors.

Consequently, if we can extract $(g_{l}^{[p]}, \theta_{l}^{[p]}, \tau_{l}^{[p]})$ from the antenna module $p$, then the receiver can infer $(g_{l}^{[q]}, \theta_{l}^{[q]}, \tau_{l}^{[q]})$ for antenna module $q$. Specifically, we use
\begin{subequations} \label{eq:thetal+taul+rho}
    \begin{align}
    \dot{\theta}_{l}^{[q]} &\leftarrow \theta_{l}^{[p]} + \Delta \theta^{[p,q]}, \label{eq:thetal+delta} \\
    \dot{\tau}_{l}^{[q]} &\leftarrow \tau_{l}^{[p]} + \Delta \tau_{l}^{[p,q]}, \label{eq:tau+delta}\\
    \dot{g}_{l}^{[q]} &\leftarrow \rho^{[p,q]} \cdot g_{l}^{[p]}, \label{eq:rho}
    \end{align}
\end{subequations}
for $l=1,\cdots, L$, where $\Delta \theta^{[p,q]}$ denotes the angular rotation of $p$ with respect to (w.r.t.) $q$, $\Delta \tau_{l}^{[p,q]}$ based on \eqref{eq:delta_tau} denotes the delay time of $p$ w.r.t. $q$, and $\rho^{[p,q]}$ is the power ratio of $p$ to $q$. An interesting observation from \eqref{eq:thetal+taul+rho} is that angular rotation $\Delta \theta^{[p,q]}$ and power ratio $\rho^{[p,q]}$ do not vary with paths, while delay time $\Delta \tau_{l}^{[p,q]}$ varies with paths. In practice, each mobile shall have different $\Delta\theta^{[p,q]}$ due to the configuration of the antenna modules or the material of the shell. Therefore, calibration is required to obtain a lookup table composed of the corresponding $\Delta \theta^{[p,q]}$.

\subsection{Antenna Switching and Beam Alignment}

We can use \eqref{eq:thetal+taul+rho} to estimate the corresponding channel of antenna module $q$ by extracting $(g_{s,l}^{[p]}, \theta_{s,l}^{[p]}, \tau_{s,l}^{[p]})$ of $L_s$ paths from antenna module $p$ for the $s$-th transmit beam as follows:
\begin{equation}
\dot{\qH}_{s}^{[q]} = \sum_{l=1}^{L_s} \dot{g}_{s,l}^{[q]} \qa(\dot{\theta}_{s,l}^{[q]}) \qb^{H}(\dot{\tau}_{s,l}^{[q]}),
\end{equation}
which we called \emph{virtual} CSI. The UE may not know the beamforming scheme and the antenna configuration of the BS, so we regard the channels formed by different transmitter beams as independent channels. According to 5G NR specifications, the BS can emit CSI-RS corresponding to different transmit beams during beam adjustment. Therefore, the UE can sequentially acquire the virtual CSIs corresponding to these transmit beams without interrupting communications. All candidate beams can be found without physically rescanning the entire receiver beams and antenna modules by reconstructing the virtual CSI matrix.

We can predict the candidate transmitter and receiver beam pairs among all the antenna modules by using the virtual CSI and calculating
\begin{equation} \label{eq:Bm}
B_{s,m}^{[p]} = \sum_{\forall f_k} \left|\dot{H}_{s,m}^{[p]}[f_k] \right|^2, ~~\forall s, m, p ,
\end{equation}
where
\begin{equation} \label{eq:tHm}
 \dot{H}_{s,m}^{[p]}[f_k] = \sum_{l=1}^{L_s} \dot{g}_{s,l}^{[p]} \cdot A_{m}{\left(\dot{\theta}_{s,l}^{[p]}\right)} \cdot e^{-j 2 \pi \dot{\tau}_{s,l}^{[p]} f_k },
\end{equation}
is the entry of $\dot{\qH}_{s}^{[p]}$. We can keep track of a small set of ``good'' beam pairs and quickly switch to the best one before the current beam fails. Beams can be switched freely among the antenna modules because
we have considered all the available beams of multiple antenna modules.

From \eqref{eq:Hm}, we observe that ToA $\dot{\tau}_{s,l}^{[q]}$ changes the phase of each path at different frequencies.
The receiver beam can only focus on one direction at a time because analog beamforming is adopted;
that is, the selected beam should be applied to all the subcarriers. We infer that
the relative ToA $\Delta \tau_{l}^{[p,q]}$ between $p$ and $q$ is irrelevant in analog beam selection.
We verify this inference through simulations in a later section. Therefore, we can change the form of the virtual CSI \eqref{eq:tHm} to
\begin{equation} \label{eq:Hm^vir_simple}
\ddot{H}_{s,m}^{[q]}[f_k] = \sum_{l=1}^{L_s} \dot{g}_{s,l}^{[q]} \cdot A_{m}(\dot{\theta}_{s,l}^{[q]}) \cdot e^{-j 2 \pi \tau_{s,l}  f_k },
\end{equation}
by assuming that all the antenna modules simultaneously receive path signals. Notably, \eqref{eq:Hm^vir_simple} differs from
\eqref{eq:tHm} only in terms of path ToAs $\tau_{s,l}$. This simplification can facilitate the implementation of
\eqref{eq:thetal+taul+rho} through a simple table lookup for AoAs and doing nothing for ToAs.

Although the underlying propagation paths for the blocked antenna module may be different because of coupling and reflection caused by fingers, the exact propagation properties of this antenna module are unimportant because selecting a blocked module should be avoided. Specifically, if an antenna module is hand blocked, then the antenna module's $\rho^{[p,q]}$ is a small value. In this case, this antenna module is not selected eventually.

\subsection{Path Parameter Estimation Algorithm}

Estimating the ToA, AoA, and gain of multipath from a channel is a critical task in realizing Fast-ABS.
This task should be conducted under beam-specific CSI measurements of different receiver beam patterns under the same antenna module. To ease notations, we focus on a transmitter beam and thus omit transmitter beam index $s$ and antenna module index $p$ from the channel-related parameters in this subsection. With this simplification, we rewrite \eqref{eq:H_mat} in a concise formulation as follows:
\begin{equation} \label{eq:sH_mat}
 \qH = \sum_{l=1}^{L} g_{l} \qa(\theta_{l}) \qb^{H}(\tau_{l}),
\end{equation}
The vectorization of $\qH$ yields
\begin{equation} \label{eq:vec_h}
\qh \triangleq {\sf vec}(\qH) = \sum_{l=1}^{L} g_{l} \qv(\theta_{l},\tau_{l}),
\end{equation}
where ${\qv(\theta,\tau) \triangleq \qb(\tau) \otimes \qa(\theta)}$.

From \eqref{eq:vec_h}, beam-specific CSI measurements can be modeled as
\begin{equation} \label{eq:vec_y}
\qy = \sum_{l=1}^{L} g_{l} \qv(\theta_{l},\tau_{l}) +\qz,
\end{equation}
where $\qz$ is the additive noise vector. The three-tuples of the paths $(\qg, \qtheta ,\qtau) \triangleq \{(g_{l}, \theta_{l} ,\tau_{l}) \}_{l=1}^{L}$
can be estimated by jointly minimizing
\begin{equation} \label{eq:costJ}
J(\qg, \qtheta ,\qtau) = \left\| \qy - \sum_{l=1}^{L} g_{l} \qv(\theta_{l},\tau_{l}) \right\|_2^2.
\end{equation}
As a solution to the aforementioned optimization problem, numerous algorithms have been developed over the past half-century. Any algorithm with high accuracy \cite{repNOMP1,repNOMP2} can be used in Fast-ABS to estimate the ToA, AoA, and gain of multipath.
In this paper, we employ the Newtonized orthogonal matching pursuit (NOMP) algorithm \cite{NOMP} in Fast-ABS, because we find that NOMP provides better performance and lower complexity than
many existing algorithms \cite{NOMP_ref}.

In NOMP, a detection estimation method is to identify each signal path and minimize \eqref{eq:costJ}
with low complexity. The strongest signal path is initially identified and
from $\qy$. Then, the weak signal path is determined using the residual signal $\qy_{\rm r} = \qy -  \hat{g} \qv(\hat{\theta} ,\hat{\tau} )$. Let
\begin{equation}
J_{\rm r}(g, \theta ,\tau) = \left\| \qy_{\rm r} -  g \qv(\theta ,\tau ) \right\|_2^2.
\end{equation}
Specifically, the above procedure consists of two stages:
\begin{enumerate}
    \item In coarse detection, the coarse estimates
    of AoA, ToA, and complex gain are obtained using pre-computed $\{ \qv(\theta ,\tau ) \}$
    \begin{subequations}
        \begin{align}
        (\hat{\theta},\hat{\tau}) &= \argmax_{\theta \in \qTheta, \tau \in \qGamma} |\qv^{H}(\theta,\tau) \qy|^2,  \label{eq:coarse_theta}\\
        \hat{g} &= {\qv^{H}(\hat{\theta},\hat{\tau}) \qy}/{ \| \qv (\hat{\theta},\hat{\tau}) \|_2^2 },
        \end{align}
    \end{subequations}
    where $\qTheta$ and $\qGamma$ denote the discretized sets of AoAs and ToAs, respectively.

    \item In refinement, estimates are refined using Newton's method:
    \begin{subequations}
        \begin{align}
        (\hat{\theta},\hat{\tau}) &\leftarrow (\hat{\theta},\hat{\tau}) - [ \nabla^2 J_{\rm r} (\hat{g}, \hat{\theta},\hat{\tau} )]^{-1} \nabla J_{\rm r} (\hat{g}, \hat{\theta},\hat{\tau}) , \\
        \hat{g} &\leftarrow  {\qv^{H}(\hat{\theta},\hat{\tau}) \qy}/{ \| \qv (\hat{\theta},\hat{\tau}) \|_2^2 },
        \end{align}
    \end{subequations}
    where $ \nabla^2 J_{\rm r} $  and $\nabla J_{\rm r}$ denote the Hessian and gradient
    of $J$, respectively, with respect to $(\theta ,\tau)$ at the current estimate $(\hat{\theta},\hat{\tau})$ provided by \eqref{eq:coarse_theta}.

\end{enumerate}
The algorithm is repeated with the residual signal $\qy_{\rm r} $ to estimate other paths. The refinement steps
are repeated after each new detection for all the paths in a cyclic manner for a few rounds to improve accuracy.

We end this subsection by remarking two points. First, ES is a common method in mmWave systems to find the best beam direction with substantial searching latency. The search space should be partitioned to search for the best beam direction. With these segmentation angles, the corresponding received SNR (RSNR) of each beam can be obtained. The performance of ES depends on the size of the fine-grained partitions of the search space. However, this approach causes unacceptable latency for analog-based systems that need to perform beam scanning and find the best beam pair. We regard the result corresponding to ES with fine-grained codebook as an oracle solution. As long as the three-tuple of the channel is available, all the candidate beams can be found without physically scanning the entire receiver beams. The NOMP algorithm can estimate the three tuples of a channel by using a small number of received beams. Simulations show that the channel reconstruction approach combined with the NOMP algorithm can achieve the oracle solution.

Second, the beam-specific CSI measurements $\qy$ in \eqref{eq:vec_y} are obtained through a set of receiver beams. Therefore, the remaining problem is how to determine proper receiver beams so that the NOMP algorithm can estimate the three tuples of a channel by using a small number of received beams. We discuss this beam selection problem later in the next section.

\subsection{Summary}
\begin{figure}
    \begin{center}
        \resizebox{2.4in}{!}{%
            \includegraphics*{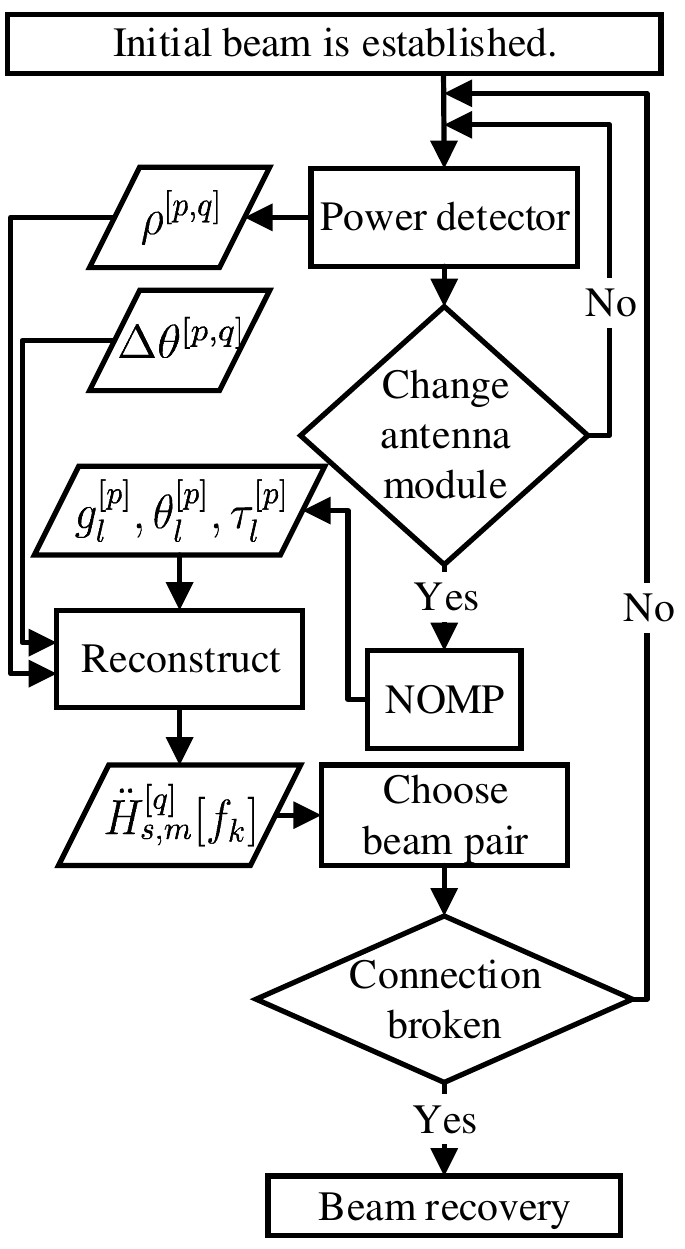} }%
        \caption{State-flow of Fast-ABS.}\label{fig:state-flow}
    \end{center}
\end{figure}

The summary of Fast-ABS is illustrated in Figure \ref{fig:state-flow}. We assume that a connection has been established during the initial access; that is, a beam pair has been initially established between the BS and the UE in the downlink directions. With Fast-ABS, the transmitter- and receiver-side beam adjustment is performed as follows. First, power detection is carried out on all antenna modules on the UE, which is used to obtain the power ratio $\rho^{[p,q]}$ between the antenna modules. Given a transmit beam from the BS, said the $s$-th beam, beam sweeping is conducted on an antenna module of the UE to measure a configured CSI-RS in sequence. Then, three tuples for each $L_s$ path are extracted by using the NOMP algorithm, and the corresponding virtual CSI \eqref{eq:Hm^vir_simple} of other antenna modules is constructed. If the link quality provided by the current antenna module is significantly lower than that of the other antenna modules, then switching to the best antenna module and its corresponding best receiver beam can be achieved with the UE. Otherwise, beam sweeping can be performed with the BS so that the UE can be employed to obtain the virtual CSI corresponding to other transmitter beams from the BS. Next, the channel conditions of the beam pairs $\{ B_{s,m}^{[p]} \}$ overall $m, p$ and available $s$ are calculated via the UE in accordance with \eqref{eq:Bm}. Switching to the best antenna module and its corresponding best transmitter and receiver beam pair can be carried out on the basis of $\{ B_{s,m}^{[p]}, \forall s,m,p \}$ via the UE. Notably, excessively violent movement and hand blockage on mobile phones are the usual causes of broken connections. When a connection is broken, a mobile phone takes a long time to restore a communication link between the BS and the UE at the beam recovery stage. Fast-ABS is mainly used to avoid frequent disrupted connections due to hand blockage.

\section{Performance Analysis}

\subsection{CRLB}

In this subsection, we analyze the lower bound of the estimation regarding the three tuples mentioned in the previous section.
Let $\qpsi = [\qpsi_1^T, \ldots, \qpsi_L^T]^T\in\mathbb{R}^{4L\times 1} $ be the vector consisting of all the unknown path parameters in which
\begin{equation} \label{eq:defpsi}
\qpsi_l = {\left[ |g_l|, \angle{g_l}, \theta_l, \tau_l \right]}^T \in \mathbb{R}^{4\times 1}
\end{equation}
consists of the parameters of the $l$-th path, where $|g_l|$ and $\angle{g_l}$ are the absolute and phase of the complex-valued channel gain, respectively. Mobile devices can only power one antenna module to estimate the three-tuple parameters at a time, said $p$, so we initially analyze the lower bound based on a single antenna module. In addition, for convenience of notation, we focus on the derivation of a single transmission beam. Using notation $\qpsi_l$ and  removing $p$ and $s$ in \eqref{eq:Hm}, we can rewrite the channel response (or CSI) as
\begin{equation}\label{eq:Hmk}
H_{m}[f_k] = \sum_{l=1}^{L} |g_{l}| e^{j\angle{g_l}} \cdot A_{m}(\theta_l) \cdot e^{-j 2 \pi \tau_l f_k }.
\end{equation}
For an unbiased estimator, estimation variance is bounded by CRLB, which is the inverse of the $4L \times 4L$ Fisher information matrix (FIM) $\qF(\qpsi)$ \cite{fisher_log} defined as
\begin{equation}\label{eq:fisher_ele}
[\qF(\qpsi)]_{u, v} = \mathbb{E}_{\qy|\qpsi}{\left\{\frac{\partial\ln p(\qy|\qpsi)}{\partial\psi_u}\frac{\partial\ln p(\qy|\qpsi)}{\partial\psi_v}\right\}},
\end{equation}
where $p(\qy|\qpsi)$ is the likelihood function of $\qy$ conditioned on $\qpsi$, and
expectation is taken over noise distribution. As shown in \eqref{eq:vec_y}, $p(\qy|\qpsi)$  follows a normal distraction with variance $\sigma^2_{z}$ and mean $\qh$.
As a result, \eqref{eq:fisher_ele} can be written as
\begin{equation}\label{eq:fisher_ele_sim}
[\qF(\qpsi)]_{u, v} = \frac{2}{\sigma^2_{z}}\Re\left\{\sum_{m=1}^{M_s}\sum_{k=1}^{N}\frac{\partial H^*_{m}[f_k]}{\partial\psi_u}\frac{\partial H_{m}[f_k]}{\partial\psi_v}\right\}.
\end{equation}
Let $\hat{\qpsi}$ be the unbiased estimator of $\qpsi$. Then, the MSE of entry $\psi_u$ is bounded as
\begin{equation}\label{eq:CRLB}
 \mathbb{E}_{\qy|\qpsi}\left\{ |\hat{\psi}_u-\psi_u|^2   \right\}
 \geq \left[\qF^{-1}(\qpsi) \right]_{u,u},
\end{equation}
where $\qF(\qpsi)$ should be nonsingular.

The $4L\times 4L$ FIM can be sliced into $L^2$ submatrices as
\begin{equation}
\qF(\qpsi) = \begin{bmatrix}
\qF(\qpsi_1, \qpsi_1)& \cdots & \qF(\qpsi_1, \qpsi_L)\\
\vdots& \ddots & \vdots\\
\qF(\qpsi_L, \qpsi_1)& \cdots & \qF(\qpsi_L, \qpsi_L)
\end{bmatrix},
\end{equation}
where ${\qF(\qpsi_l, \qpsi_k) =  \Re\{ ({\partial H^*_{m}[f_k]}/{\partial \qpsi_l}) ({\partial H_{m}[f_k]}/{\partial \qpsi_k}) \} }$ is a $4 \times 4$ matrix.
The rank deficiency of $\qF(\qpsi)$ arises if two or more paths have an extremely close delay and angle, which cause the determinant of $\qF(\qpsi)$ to approach $0$. Fortunately, the number of paths to the receiver is very small in mmWave systems because of reflection difficulty in a high-frequency band. Moreover, paths are characterized by separation in delay because of a large bandwidth \cite{path}. Even if delays and angles are very close in a few paths, these similar paths can be replaced with a single path whose amplitude is the sum of the amplitudes of these paths.
On the basis of this argument, we assume that $\qF(\qpsi)$ is nonsingular.
Thus, $\qF(\qpsi)$ is transformed into a block diagonal matrix. Each submatrix on the diagonal of $\qF(\qpsi)$ can be written as
\begin{equation}\label{eq:singlepath_F}
\qF(\qpsi_l, \qpsi_{l})
=\left[
\begin{array}{cccc}
\qF(|g_l|, |g_{l}|) & \qF(|g_l|, \angle g_{l}) & \qF(|g_l|, \theta_{l}) & \qF(|g_l|, \tau_{l})\\
\qF(\angle g_l, |g_{l}|) & \qF(\angle g_l, \angle g_{l}) & \qF(\angle g_l, \theta_{l}) & \qF(\angle g_l, \tau_{l})\\
\qF(\theta_l, |g_{l}|) & \qF(\theta_l, \angle g_{l}) & \qF(\theta_l, \theta_{l}) & \qF(\theta_l, \tau_{l})\\
\qF(\tau_l, |g_{l}|) & \qF(\tau_l, \angle g_{l}) & \qF(\tau_l, \theta_{l}) & \qF(\tau_l, \tau_{l})
\end{array}\right].
\end{equation}
Considering that the inverse of a block diagonal matrix is a block diagonal matrix with the inverse of the original blocks on its diagonal, we can separately calculate the inverse matrix of each submatrix.
Therefore, we focus on FIM in the single-path scenario, which is denoted as $\qF(\qpsi)$ with $\qpsi = [|g|, \angle{g}, \theta, \tau]^T$ in the following description. The entries of $\qF(\qpsi)$ are derived in Appendix A.

The CRLB provides the lower bound on the variance of the estimated parameters. We are also interested in the error of the reconstructed CSI defined in \eqref{eq:Hmk}. The MSE of the estimated CSI $\hat{H}_{m}[f_k]$ of the actual channel $H_{m}[f_k]$ is defined as
\begin{equation}
\mathrm{MSE}_m(\qpsi)[f_k] = \mathbb{E}\{|H_{m}[f_k]-\hat{H}_{m}[f_k]|^2\},
\end{equation}
where the expectation is taken over the noise for the deterministic $\qpsi$. The FIM of the CSI can be determined through a transformation of variables from $\qpsi$ to $H_{m}[f_k]$.
That is, the lower bound of the CSI is obtained via the transformation vector ${\partial H_{m}[f_k]}/{\partial \qpsi}$ as
\begin{equation}\label{eq:LB}
\mathrm{LB}_m(\qpsi)[f_k] = \left(\frac{\partial H_{m}[f_k]}{\partial \qpsi}\right)^H \qF^{-1}(\qpsi)\frac{\partial H_{m}[f_k]}{\partial \qpsi},
\end{equation}
where the entries of ${\partial H_{m}[f_k]}/{\partial \qpsi}$ are given by
\begin{subequations} \label{eq:dev_F}
\begin{align}
\frac{d H_{m}[f_k]}{d |g|} &= e^{j\alpha}  A_{m}(\theta),
&&\frac{d H_{m}[f_k]}{d \angle g} = |g| je^{j\alpha} A_{m}(\theta),\\
\frac{d H_{m}[f_k]}{d \theta} &= |g| e^{j\alpha} {A}_{m}'(\theta),
&&\frac{d H_{m}[f_k]}{d \tau} =|g| e^{j\alpha}  A_{m}(\theta) \left(-j 2 \pi f_k \right),
\end{align}
\end{subequations}
with $ \alpha = \angle{g} -2 \pi \tau f_k$ and ${A}_{m}'(\theta) = {d A_{m}(\theta)}/{d\theta}$. After these partial derivatives are inserted for a specific beamforming gain $A_{m}(\theta)$, the CSI error \eqref{eq:LB} can be constructed.

\subsection{Selection of an Analog Beam Codebook}

For analog-based beamforming technology, a commonly used codebook in \eqref{eq:Am} is in the form of
\begin{equation}\label{eq:w_m}
w_m(n) = \frac{1}{\sqrt{N}}e^{j\pi n\cos\phi_m}, ~n=0, \cdots, N-1,
\end{equation}
where $\phi_m$ represents the rotation angle of the $m$-th beam. For example, if an array antenna with a uniform separation of a half wavelength ($N = 4$) is assumed and $\phi_m = 90^\circ$, then the beam pattern $A_m(\theta)$ is obtained by substituting $\theta=0$ to $\pi$ in \eqref{eq:Am}. We derive the corresponding beam pattern in a Cartesian plot (Figure \ref{fig:sinc}).
Here, we adopt the codebook in \eqref{eq:w_m} for ease of understanding. The principle obtained below can be extended to arbitrary codebooks.

\begin{figure}
    \begin{center}
        \resizebox{2.5in}{!}{%
            \includegraphics*{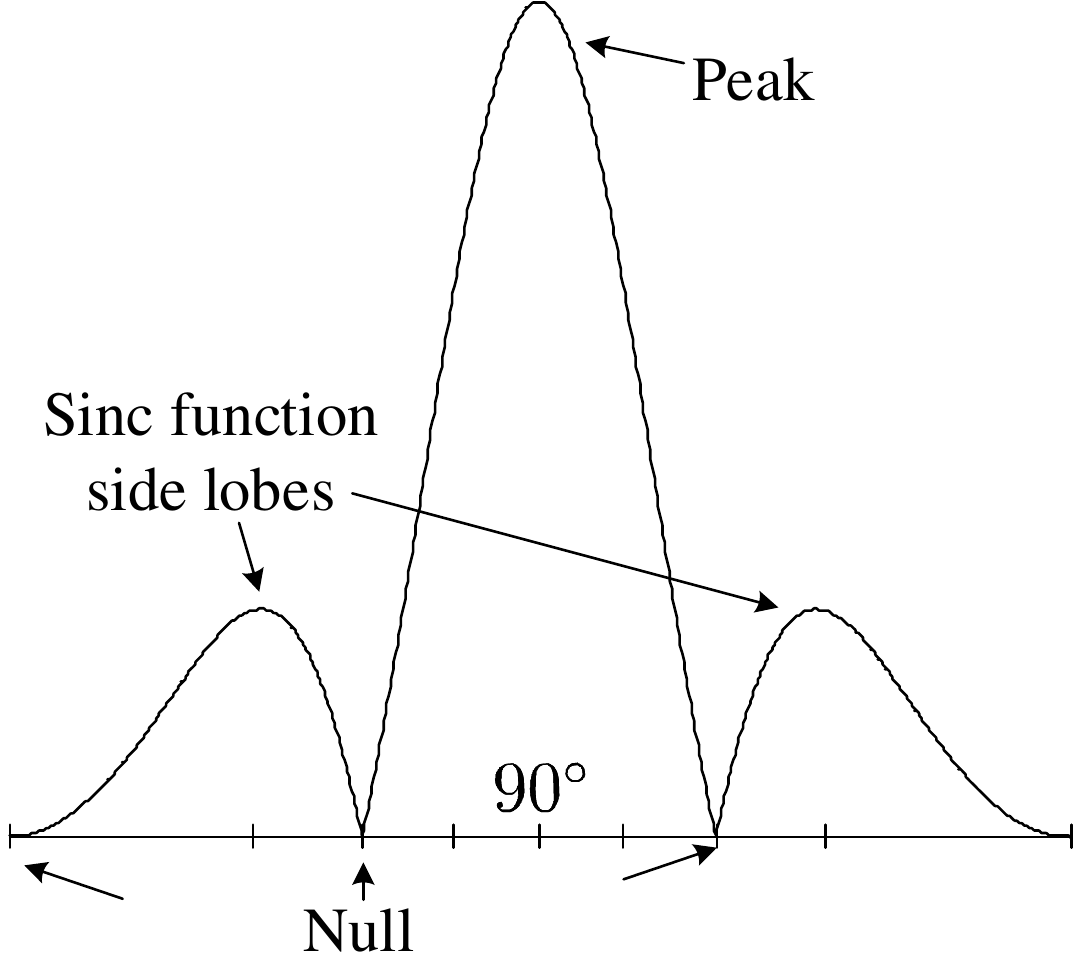} }%
        \caption{Sinc function of an analog beam.}\label{fig:sinc}
    \end{center}
\end{figure}

\begin{figure}
    \begin{center}
        \resizebox{3.5in}{!}{%
            \includegraphics*{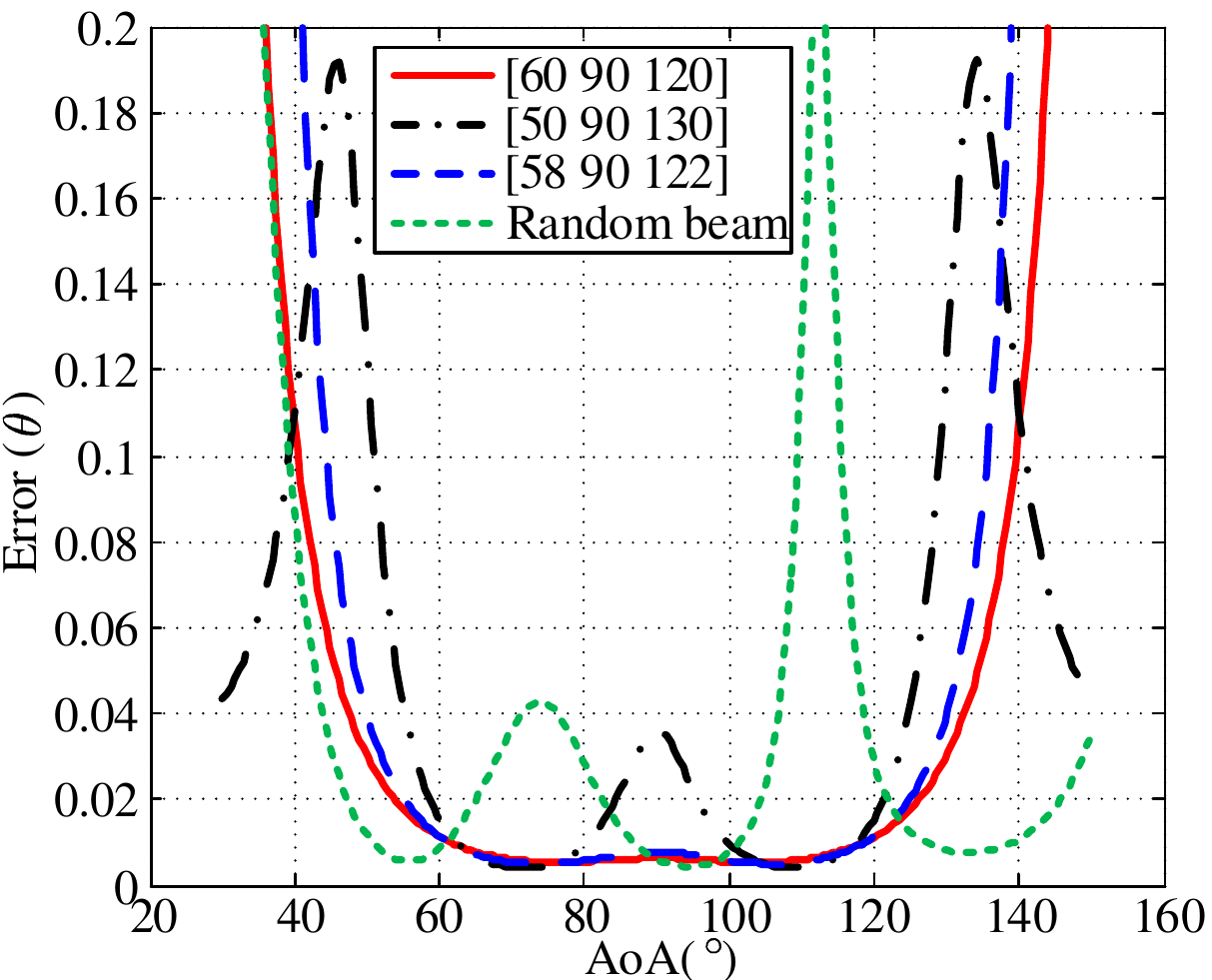} }%
        \caption{CRLB of different beam codebooks.}\label{fig:CRLB_theta}
    \end{center}
\end{figure}

The NOMP algorithm estimates the three tuples of a path through beam-specific CSI measurements.
Beam directions can be determined by partitioning a search space into many fine-grained beams.
However, doing so causes unacceptable latency for analog-based systems that need to perform beam sweeping.
Therefore, a proper set of $\phi_m$ for beam sweeping should be determined so that the three tuples of a channel can be obtained using a small number of received beams. Figure \ref{fig:sinc} shows that the maximal information regarding a channel can be obtained at $\theta = 90^\circ$. Moreover, the beam still provides useful information about $\theta$ except for the null points at $\theta = 60^\circ$, $120^\circ$ ,and $180^\circ$. Therefore, angles should be probed at the null points to obtain the full angular observation. In this way, the individual peaks of a beam pattern line up with the nulls of the other beam pattern to achieve no interference between beam patterns.
Therefore, a channel with few received beams can be efficiently estimated by using a beam codebook in which beams do not interfere with one another. This characteristic is referred to as the \emph{orthogonal beam principle}.
In fact, this principle is close to the spectrum characteristic of an OFDM system, where orthogonally spaced subcarriers provide  the most efficient way to utilize frequency bands without creating a subcarrier interference.

The CRLB derived in the previous subsection is exploited to show the quality of the beam codebook.  From \eqref{eq:defpsi} and \eqref{eq:CRLB}, the CRLB of $\theta$ is given by $\sigma^2_{\theta} \triangleq \left[\qF^{-1}(\qpsi) \right]_{3,3}$, whose details are shown in following proposition.
\begin{Proposition}
    The CRLB of $\theta$ can be simplified as
    \begin{equation} \label{eq:CRLB_theta}
    \sigma^2_{\theta} = \frac{\sigma^2_{z}A}{2N_s|g|^2(AD-\Re\{B C\})},
    \end{equation}
    where parameters $A$, $B$, $C$, and $D$ are defined by \eqref{eq:AtoD} in Appendix A.
\end{Proposition}
\begin{proof}
    The CRLB of $\theta$ can be calculated from \eqref{eq:sig_theta} in Appendix A. We can infer from \eqref{eq:AtoD} that $\Im\{A\} = 0$ and $\Im\{D\} = 0$. Therefore, the parameter $J$ which is defined in \eqref{eq:J} can be simplified as
    \begin{align}
    J &= A^2D+\Im\{B\}\Im\{C\}A-A\Re\{B\}\Re\{C\}\notag \\
    &=A(AD-\Re\{BC\}).\label{eq:J_sim}
    \end{align}
    Substitute \eqref{eq:J_sim} into \eqref{eq:sig_theta}, the simplification can be obtained as \eqref{eq:CRLB_theta}.
\end{proof}

Our focus is on the relationship of the beam codebook with the CRLB of $\theta$, so the above equation is simplified to
\begin{equation}\label{eq:CRLB_theta_sim}
\sigma^2_{\theta} \varpropto \frac{A}{AD-\Re\{B C\}} \triangleq {\rm Error}(\theta).
\end{equation}
Thus, \eqref{eq:CRLB_theta_sim} is used to examine the quality of the beam codebook. For example, an antenna module has four antenna elements with a uniform separation of a half wavelength, and the system only allows three beams to receive signals. In accordance with the orthogonal beam principle, the three-beam codebook should be $\phi_m = \{ 60^\circ, 90^\circ, 120^\circ \}$. Figure \ref{fig:CRLB_theta} shows that the angular estimation error of the codebook found on the basis of the orthogonal beam principle is lower than those of the codebook composed of angles $\{ 50^\circ, 90^\circ, 130^\circ \}$ and a random codebook over the whole AoA range. Figure \ref{fig:CRLB_theta} illustrates that adding a reasonable error value on the angular error, i.e., the case of $\{ 58^\circ, 90^\circ, 122^\circ \}$, barely affects the estimation error.
Plugging different numbers and sets of $\phi_m$ for beam sweeping into \eqref{eq:CRLB_theta_sim} can obtain the lower bound of the estimated $\theta$. Hence, we can use \eqref{eq:CRLB_theta_sim} to directly find the effective and non-wasteful codebook set for beam sweeping.

\subsection{Error of Virtual CSI}

In \eqref{eq:LB}, the CSI error for the  power one antenna module is shown. Now, the virtual CSI error for other antenna modules is considered. From \eqref{eq:Hm^vir_simple}, the virtual CSI of $q$ can be reconstructed as $\ddot{H}_{m}^{[q]}[f_k]$ by using the three-tuple parameters estimated from $p$.  The major difference between $p$ and $q$ is in $A_{m}(\dot{\theta}_{s,l}^{[q]})$, which is only related to the AoA and the rotation angle of the analog beam. The lower bound of the MSE of the virtual CSI at $q$ can be set as $\mathrm{LB}_{m}^{[q]}(\qpsi)[f_k]$ by using $\ddot{H}_{m}^{[q]}[f_k]$ in accordance with \eqref{eq:LB}. The AoA of the $l$-th path and the rotation angle of the analog beam at $p$ and $q$ are written as $\{\theta_{l}^{[p]}, \phi_m^{[p]}\}$ and $\{\theta_{l}^{[q]}, \phi_m^{[q]}\}$. The angle of different antenna modules is only \emph{definitionally} different. The relationship between $\theta_{l}^{[p]}$ and $\theta_{l}^{[q]}$ can be determined explicitly by using \eqref{eq:thetal+delta} through the angular rotation $\Delta \theta^{[p,q]}$.
In addition, the relationship between $\phi_m^{[p]}$ and $\phi_m^{[q]}$ can also be written as
\begin{equation}\label{eq:phi}
\phi_m^{[q]} = \phi_m^{[p]} + \Delta \theta^{[p,q]}.
\end{equation}
Consequently, we obtain
\begin{equation}
\mathrm{LB}_{m}^{[q]}(\qpsi)[f_k] = \mathrm{LB}_{m}(\qpsi)[f_k].
\end{equation}
That is, the CSI of the powered antenna module can be utilized to determine the virtual CSI error of the other antenna modules.
Therefore, the result proves that using the virtual channel for switching in Fast-ABS can achieve the same performance as that obtained with the original powered antenna module.

\section{Simulations and Implementation}

We conduct simulations and implementation to verify the efficiency of Fast-ABS. We first apply simulations to analyze the performance of Fast-ABS. Next, we discuss how to implement Fast-ABS on software-defined radio and evaluate its performance via over-the-air (OTA) tests. Regardless of simulations or experiments, we consider an antenna module of UE has four antenna elements with a uniform separation of a half wavelength. We also consider that the observable AoAs of each module are between $30$ and $150$ degrees because of the performance limitation of the antenna array at the beam edges, as observed through our measurement. The AoAs of an antenna module outside the observable range belong to the observable AoAs of the other antenna module.
The codebook received by Fast-ABS has only nine grids on the azimuth angle; that is, $\phi_m$ in \eqref{eq:w_m} is only available at $[30^{\circ}, 150^{\circ}]$ for every $15^\circ$. To obtain the performance benchmark, we perform the ES of the azimuth angle with ${M_{\rm ES} = 481}$ grids in the angular span of $[30^{\circ}, 150^{\circ}]$ with $0.25^\circ$ spacing per partition. With these segmentation angles, we can obtain the corresponding RSNR of each angle, and we set the angle corresponding to the peak value as the resulting angle of the oracle.

\subsection{Simulations}
We utilize four-, three-, and two-beam patterns to obtain beam-specific CSI measurements, that is, $M_{\rm ABS} = 2, 3, 4$.
The angular directions of the four-, three-, and two-beam codebooks are respectively centralized at $\{45^\circ, 75^\circ, 105^\circ, 135^\circ\}$, $\{60^\circ, 90^\circ, 120^\circ\}$, and $\{60^\circ, 120^\circ\}$. Beam angles in a codebook do not affect one another because of the orthogonal beam principle. In accordance with \eqref{eq:CRLB_theta}, we know that the lower bound of the estimated angle is related to SNR and the number of pilots of CSI-RS $N_s$, given a fixing codebook. Figure \ref{fig:LB_SNR} shows the estimation errors of the four-, three-, and two-beam codebooks.
In addition, we compared the two possible pilot lengths, that is, $N_s = 300$ and $N_s = 825$. We see that the estimation accuracies improve to a greater extent by increasing beam angles in a codebook than by increasing pilot lengths.

\begin{figure}
    \begin{center}
        \resizebox{3.50in}{!}{%
            \includegraphics*{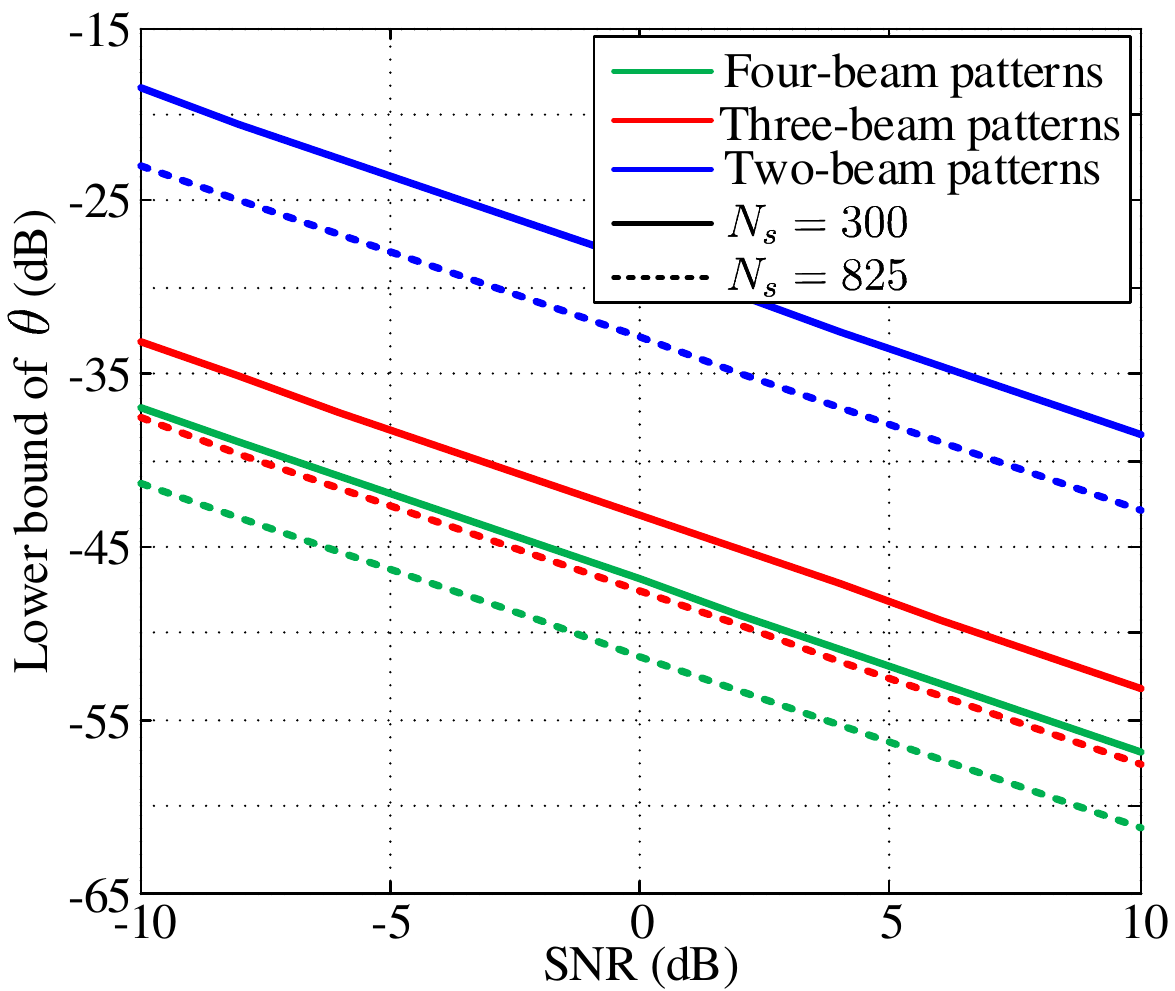} }%
        \caption{CDFs of Fast-ABS's RSNRs w.r.t. ES directly performed on module 2.}\label{fig:LB_SNR}
    \end{center}
\end{figure}

To clearly understand the performance of the considered algorithms with the used codebooks, we compare the performance of MSE by NOMP, OMP, and ES and the lower bound of the estimated angle under a fixed $N_s= 300$ and ${\rm SNR} = 10$\,dB. We consider two-path scenarios: one is the line-of-sight (LoS) path, and the other reflection path is 10 dB (case I) and 20 dB (case II) weaker than the LoS path because of reflection loss.
The phase of the complex channel gain is uniformly over $0$ and $2\pi$. ToAs are generated randomly from a distance of $0$\,m to $60$\,m. The LoS path distance is smaller than the NLoS path distance.
Figure \ref{fig:MSE_theta}(a) compares the performance of the dominant path by different algorithms when the four-beam pattern is used under case I. We can observe that the stable result of NOMP is superior to the estimation result of OMP because of granularity. Moreover, the estimation result of NOMP by using four-beam patterns is better than that of ES. This result suggests that NOMP exhibits robustness to multipath and produces a high resolution that exceeds $0.25^\circ$ and significantly reduces latency by nearly 120 times.
Figure \ref{fig:MSE_theta}(b) compares the performance of the second path by different algorithms when the four-beam pattern is used under different cases. The figure shows that OMP cannot estimate the angle of the second path and is thus unavailable for establishing virtual channels. In addition, the performance in case I is better than that in case II. The reason is obviously related to the power of the reflection path.
Figure \ref{fig:MSE_theta}(c) compares the performance of the dominant path by NOMP when using different beam patterns under different cases. The result of Figure \ref{fig:MSE_theta}(c) shows that the performance under case II is better than that under case I. This result is reasonable because we can regard case II as a single-path channel as the path difference reaches 20 dB. This result also shows that the trend estimated by NOMP is the same as that estimated by CRLB. Notably, the CRLB of $\theta$ is associated with single-path scenarios.

\begin{figure*}
   \begin{center}
       \resizebox{6.95in}{!}{%
           \includegraphics*{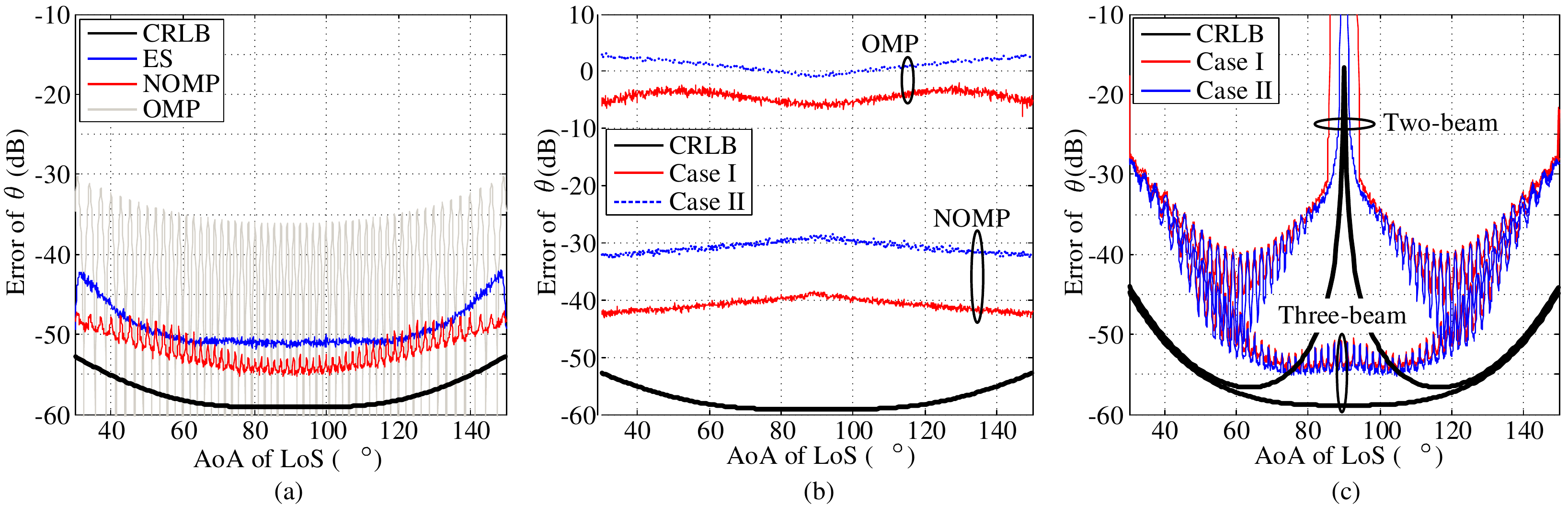} }%
       \caption{(a) Error of the dominant path angle with different algorithms. (b) Error of the second path angle with different algorithms in different case. (c) Error of the dominant path angle with different beam patterns in different case. }\label{fig:MSE_theta}
   \end{center}
\end{figure*}

The RSNR calculated after beamforming is our final performance target. Therefore, the CDF of the RSNR of the concerned methods is presented in Figure \ref{fig:NOMP_RSNR}. The parameter of AoA is generated in a random uniform distribution between $[30^{\circ}, 150^{\circ}]$. Notably, the received beam to calculate RSNR is the closest direction which is from nine grid angles, not the direction that matches the strongest path. We find that NOMP with only three- to four-beam patterns shows a similar RSNR that is optimal in all scenarios. However, the performance of the two-beam pattern is unproductive. Figure \ref{fig:MSE_theta} illustrates a range of angles, i.e., $(85^{\circ}, 95^{\circ})$, which results in an unacceptable RSNR (Figure \ref{fig:NOMP_RSNR}) for $8\% $ of cases.

\begin{figure}
   \begin{center}
       \resizebox{3.5in}{!}{%
           \includegraphics*{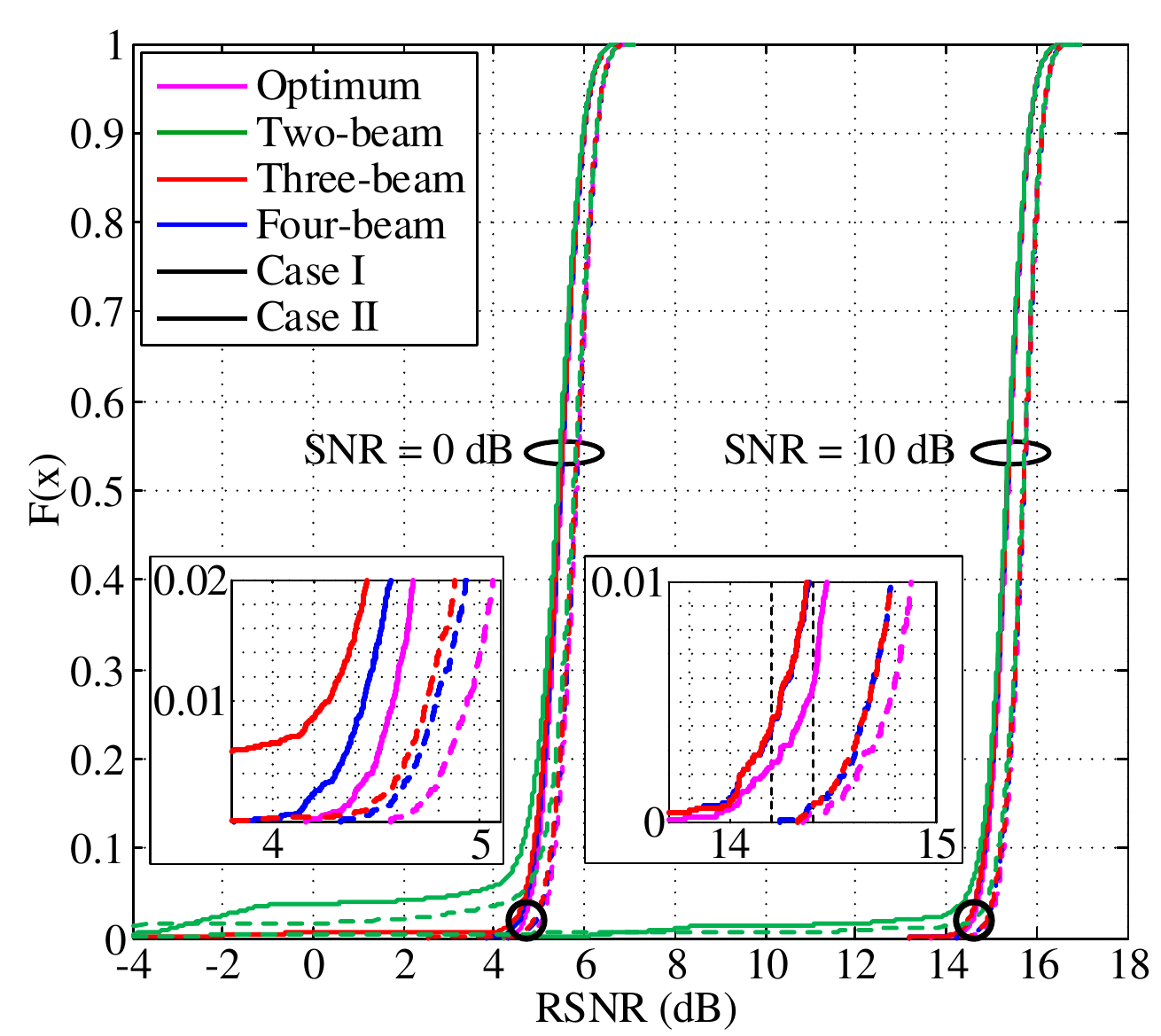} }%
       \caption{CDFs of NOMP's RSNRs w.r.t. the optimum for cases 1 and 2.}\label{fig:NOMP_RSNR}
   \end{center}
\end{figure}

In Section III.C,  the relative propagation delay between two arrays is unrelated to analog beam selection.
To verify this claim, we conduct the same simulation as in case 1 for antenna module 1, while we generate the AoAs and ToAs of antenna module 2 by using \eqref{eq:thetal+taul+rho}
with the propagation delay given by \eqref{eq:delta_tau}.
We consider that the two antenna modules are placed over the long/left and top/short edges, as shown in Figure \ref{fig:Intuition}(d).
Antenna module 1 is the current receiving antenna, while the channel and the beam of antenna module 2 are predicted by Fast-ABS.
The optimum RSNRs of module 2 are achieved by directly performing result on module 2, and the results serve as the benchmark.
Figure \ref{fig:twoarray_RSNR} illustrates that the RSNRs of Fast-ABS are the same as those of optimum.
Although propagation delays between antenna modules 1 and 2 are ignored,
the beam selected by Fast-ABS coincides entirely with the optimum across the test scenarios.
Consequently, we can use a (simple) virtual channel given by \eqref{eq:Hm^vir_simple} for beam selection.

\begin{figure}
    \begin{center}
        \resizebox{3.50in}{!}{%
            \includegraphics*{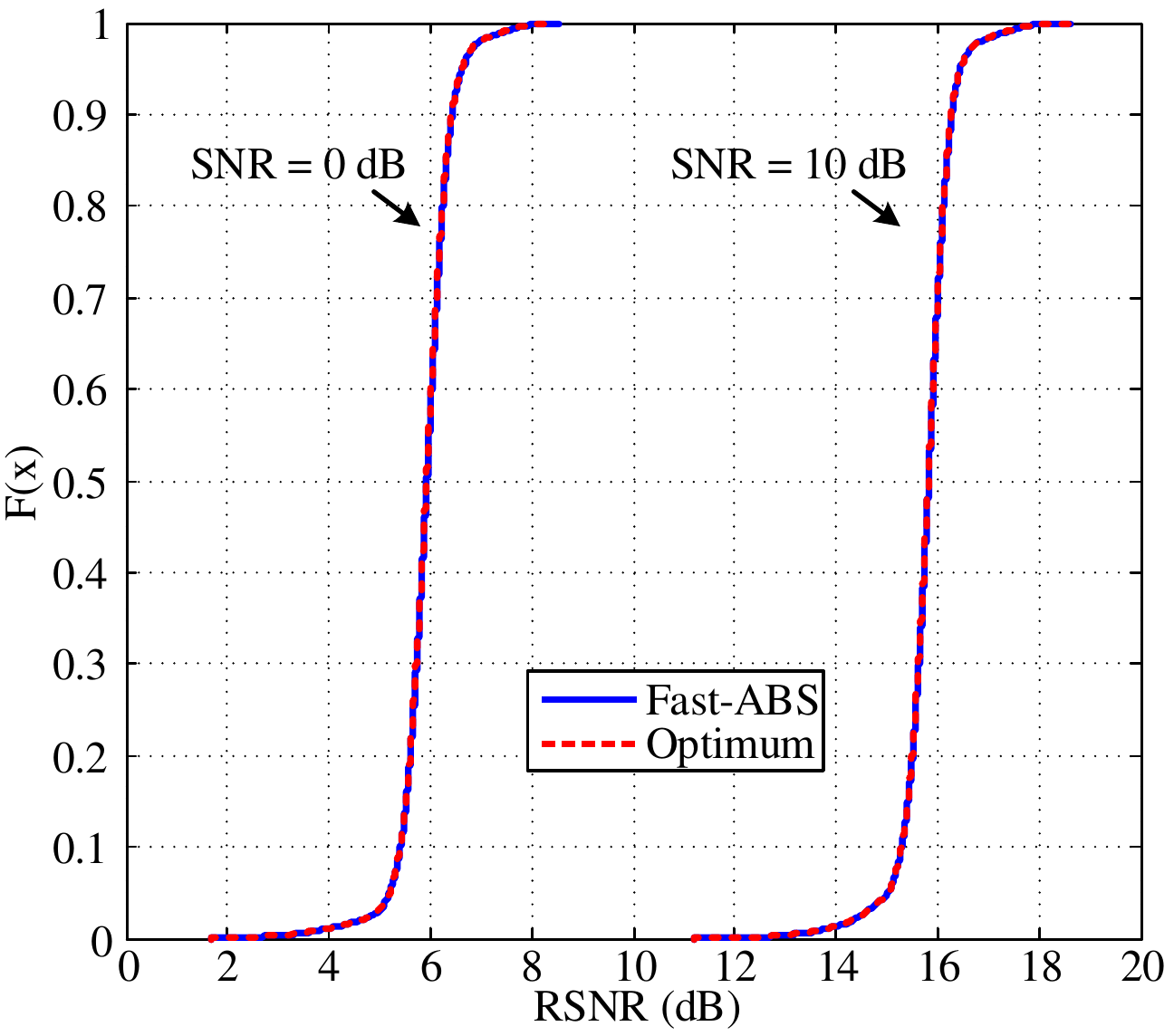} }%
        \caption{CDFs of Fast-ABS's RSNRs w.r.t. the optimum directly performed on module 2.}\label{fig:twoarray_RSNR}
    \end{center}
\end{figure}

\subsection{Implementation}
With the abovementioned simulations, we verify our proposals: 1) NOMP can be an efficient algorithm to extract the three tuples of paths through only three to four probings of beam-specific CSI measurements, and 2) a virtual channel is an effective CSI for beam selection. To verify the feasibility of Fast-ABS in practical scenarios, we propose our main rationale that the underlying physical signal paths traversed by each subarray should remain the same. That is, a fixed angular rotation occurs between two arrays on a mobile phone. Thus, we implement Fast-ABS on software defined radios and integrate it into a 5G NR physical layer.

The architecture of our testbed is shown in Figure \ref{fig:testbed_arch}. In the transmitter side, the baseband (BB) signal is first modulated to 4 GHz as an intermediate frequency (IF) signal. The BB signal is based on the 5G NR physical layer with a bandwidth of 100 MHz for a subcarrier spacing of 60 kHz. Next, an upconverter performs the quadruple operation of a 6 GHz carrier generated by a local oscillator (LO) with a frequency multiplier to obtain a 24 GHz carrier. Finally, a mixer is used to modulate a 4 GHz IF signal with a 24 GHz carrier to generate a 28 GHz radio frequency (RF) signal, and the transmit antenna is utilized to emanate the RF signal. The transmit antenna is either a horn antenna or an ${8 \times 8}$ phased array, which depends on the purpose of the following experiment. On the receiver side, the downconverter at the receiver also operates with the LO to obtain the 4 GHz IF signal back. An oscilloscope demodulates the signal to a BB signal. For our hardware, each receiver antenna element is equipped with an RF chain. We develop a calibration procedure to achieve the phase coherence of each subarray. Therefore, we can apply digital beamforming to simulate the equivalent result of the analog beamforming as in \eqref{eq:Am}.

\begin{figure*}
    \begin{center}
        \resizebox{5.30in}{!}{%
            \includegraphics*{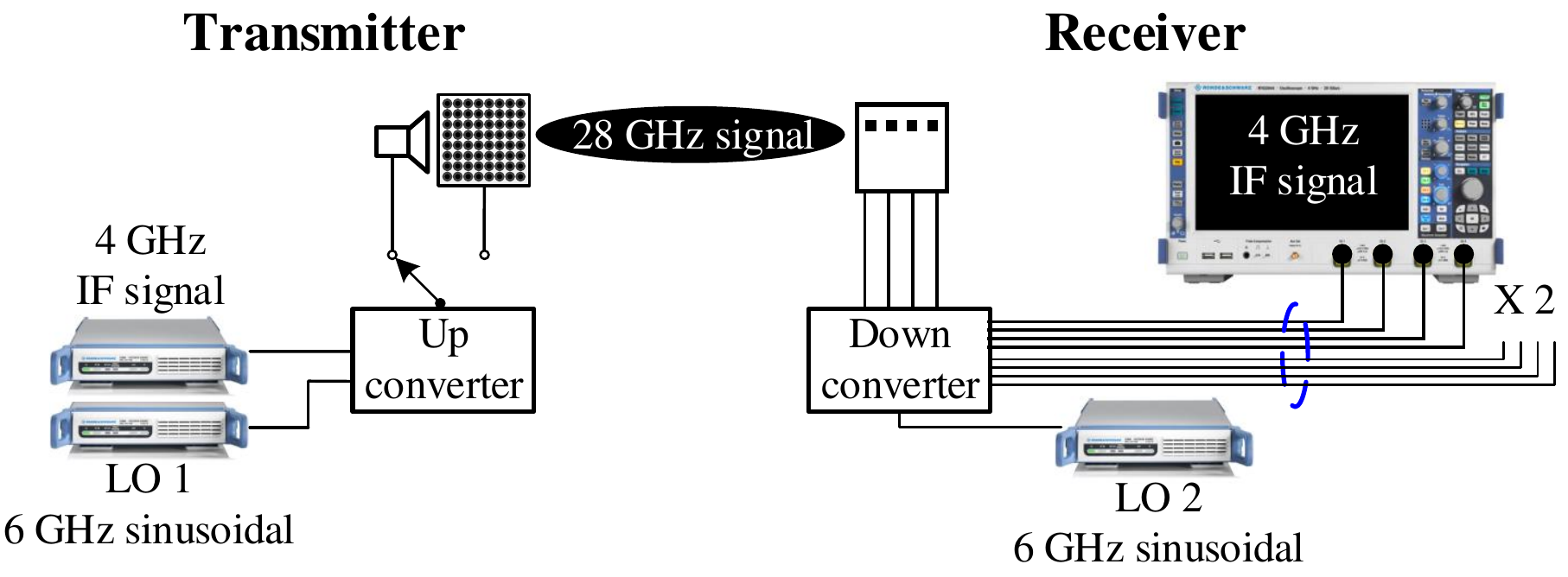} }%
        \caption{Architectural schematics of a mmWave testing platform.}\label{fig:testbed_arch}
    \end{center}
\end{figure*}

With our testbed, signals from the two antenna modules can be received \emph{simultaneously}. Thus, the testbed can be used to verify whether or not a fixed angular rotation occurs between two arrays. To clean the transmit beam, a horn antenna is used as a transmit antenna to verify this characteristic, with the distance between the transmitter and the receiver being 50\,cm, as shown in Figure \ref{fig:testbed_env}(b).
The receiver is equipped with a mobile phone consisting of two $4 \times 1$ dipole arrays placed at the two edges of the mobile phone named modules 1 and 2, as presented in Figure \ref{fig:testbed_env}(e). Figure \ref{fig:array_arch}(a) depicts the effective angular range of the two antenna modules. The schematic scenarios of LoS and NLoS are illustrated in Figure \ref{fig:array_arch}(b). In the LoS scenario, the BS is placed at the protractor, as displayed in Figure \ref{fig:testbed_env}(b), and it emits signals from $30^\circ$ to $60^\circ$ of module 2. We record every measurement result per $6^\circ$ change in position. In the NLoS scenario, the horn antenna does not directly transmit a signal to the receiving direction; instead, it utilizes an iron plate to simulate a reflection path. We rotate the iron plate to obtain different reflection paths. The measurement results are summarized in Table \ref{tab:ant_angle}; in particular, the AoAs of the two arrays are extracted simultaneously from the four-beams patterns ($M_{\rm ABS}=4$) via NOMP. The AoAs are selected from an antenna module with the strongest channel gain to represent the main AoA of the antenna module. The corresponding RSNRs are also listed in Table \ref{tab:ant_angle} to reflect the measurement quality received by that AoA direction. Approximately $90^\circ$ angular rotation is found between the two arrays in LoS and NLoS scenarios. Consequently, a fixed angular rotation is observed between the two antenna modules, and Fast-ABS is feasible.

\begin{figure*}
    \begin{center}
        \resizebox{6.0in}{!}{%
            \includegraphics*{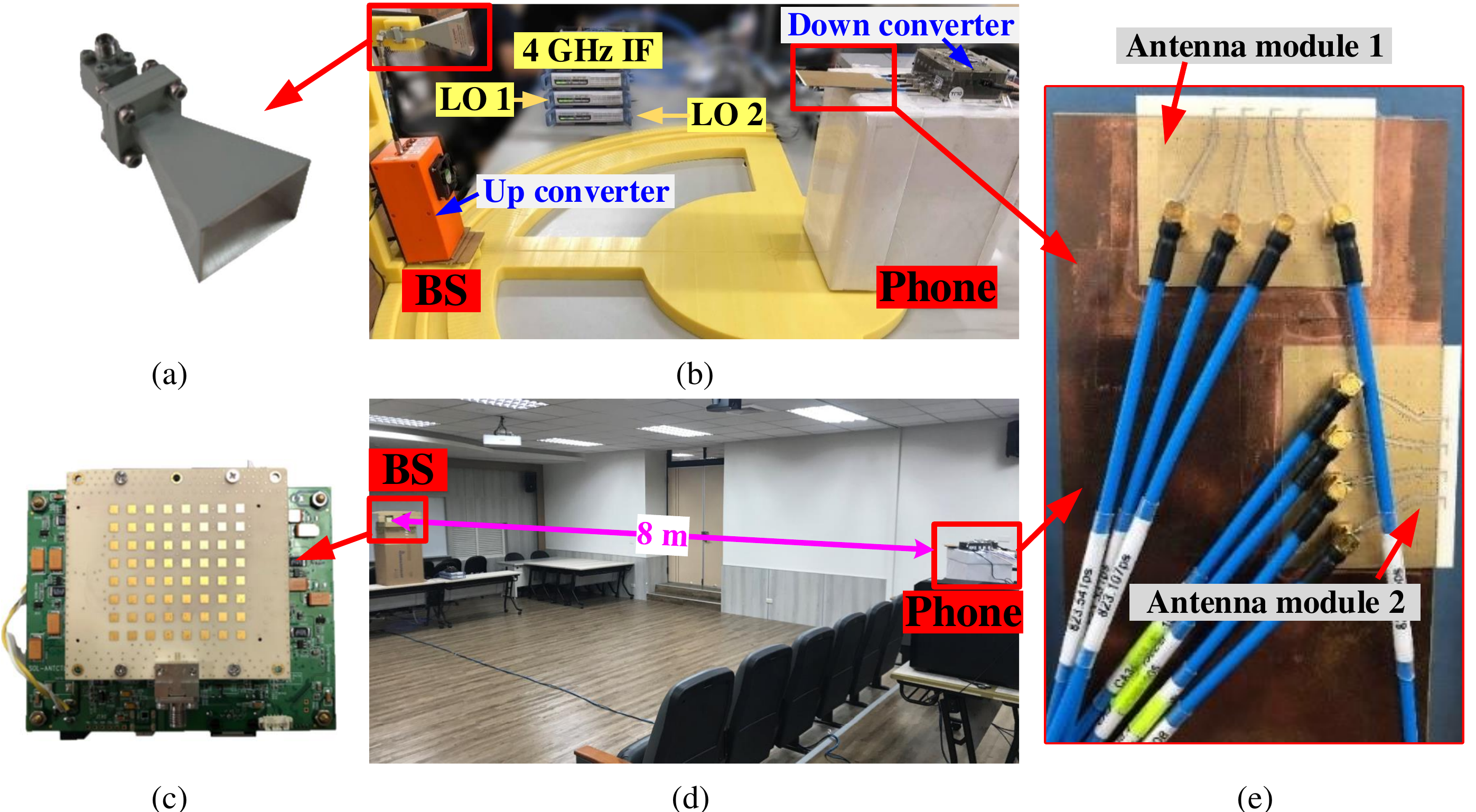} }%
        \caption{(a) Horn antenna.
        (b) Actual scenario of the mmWave testing platform with a horn antenna in the protractor.
        (c) An $8\times 8$ planar antenna array.
        (d) Long-range actual scenario for the mmWave testing platform with $8\times 8$ planar antenna array.
        (e) Two $4\times 1$ dipole arrays.}\label{fig:testbed_env}
    \end{center}
\end{figure*}

\begin{figure}
    \begin{center}
        \resizebox{3.80in}{!}{%
            \includegraphics*{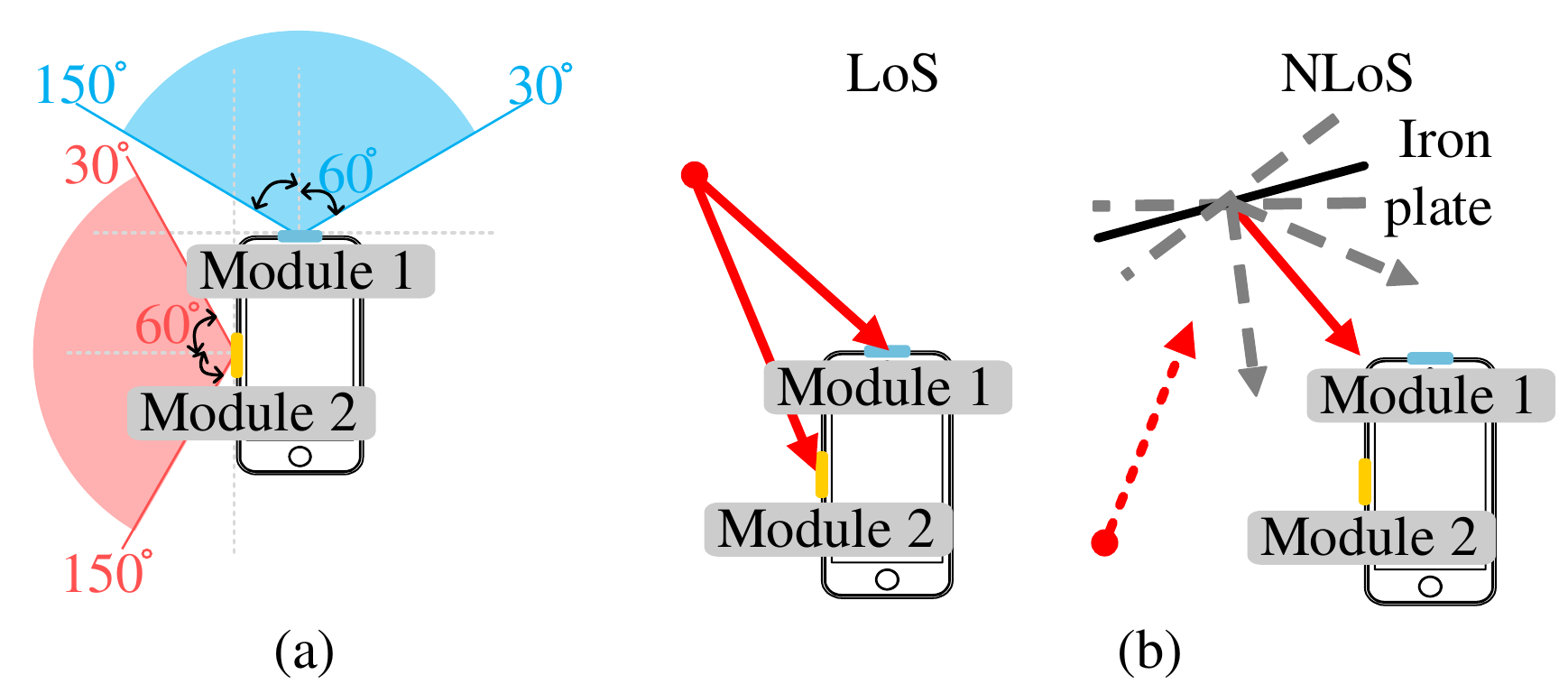} }%
        \caption{(a) Two $4\times 1$ dipole arrays. (b) Effective range of beam receiving. }\label{fig:array_arch}
    \end{center}
\end{figure}

\begin{table}
    \centering
    \begin{footnotesize}
        \caption{Two sets of subarray measurement results.}\label{tab:ant_angle}
        \begin{tabular}{cccccc}\toprule
            & \multicolumn{2}{c}{Antenna module 1} & \multicolumn{2}{c}{Antenna module 2} & \multirow{2}{*}{\begin{tabular}[c]{@{}c@{}}Difference \\ of AoA\end{tabular}} \\ \cmidrule{2-5}
            & AoA & RSNR & AoA & RSNR &  \\ \midrule
            \multirow{5}{*}{LoS} & 120.2$^\circ$ & 28.0 dB & 32.5$^\circ$ & 27.6 dB & 88$^\circ$ \\
            & 128.5$^\circ$ & 28.8 dB & 41.9$^\circ$ & 28.5 dB & 87$^\circ$ \\
            & 132.0$^\circ$ & 28.2 dB & 43.2$^\circ$ & 28.7 dB & 89$^\circ$ \\
            & 139.2$^\circ$ & 28.9 dB & 50.0$^\circ$ & 28.6 dB & 89$^\circ$ \\
            & 143.8$^\circ$ & 28.4 dB & 52.3$^\circ$ & 28.0 dB & 91$^\circ$ \\ \midrule
            \multirow{3}{*}{NLoS} & 130.5$^\circ$ & 11.8 dB & 41.1$^\circ$ & 18.5 dB & 89$^\circ$ \\
            & 134.2$^\circ$ & 14.9 dB & 45.0$^\circ$ & 18.1 dB & 89$^\circ$ \\
            & 139.0$^\circ$ & 16.0 dB & 50.8$^\circ$ & 24.2 dB & 88$^\circ$ \\ \bottomrule
        \end{tabular}
    \end{footnotesize}
\end{table}

Finally, Fast-ABS is tested under a long-range scenario in which the distance distance is 8\,m by using the $8 \times 8$ phased array, as shown in Figure \ref{fig:testbed_env}(d).
In Figure \ref{fig:testbed_env}(c), the phased array counts for eight patch antennas in horizontal and vertical directions lead to a total of $8 \times 8$ antenna elements. The codebook table of the phased array is built with $11$ beams in the horizontal and vertical directions.
The receiver also uses the mobile phone with two antenna modules as in Figure \ref{fig:testbed_env}(e).

The RSNR of Fast-ABS is cut with an approximated $200$\,s trajectory (Figure \ref{fig:Imple_result}). The results of two antenna modules can be recorded simultaneously, so the detailed performances are recorded in Table \ref{tab:Imple_result}. The ES scheme is taken for comparison. If the same antenna module is used in Fast-ABS and ES, then the RSNR performance achieved by ES is the benchmark. Our UE is equipped with a power detector on each antenna module so the UE knows the power information of all the antenna modules. Therefore, if the ES scheme is employed, the UE can decide to switch to another antenna module based on power information.
However, when doing so with the ES scheme, the UE cannot receive data immediately because the UE has to scan the entire receiver beams to establish a communication link.

Figure \ref{fig:Imple_result} illustrates the RSNR performances between Fast-ABS and ES under different scenarios, which are classified into five stages (stages I--V). 
The corresponding time complexities of each stage are summarized in Table \ref{tab:complex}.
At stage I, module 1 is utilized as the initial antenna module with the UE. The transmit beam direction is determined by the initial access. No significant RSNR difference is observed between Fast-ABS and ES at this stage. The RSNR performance of ES is slightly better than that of Fast-ABS because of the fine-grained beam directions of the former. At stage II, a power detector combined with a reconstructed channel is utilized in Fast-ABS to determine the better antenna module and the receiver beam direction. Since module 2 is better than module 1 in this scenario, Fast-ABS switches to module 2 at stage II. 
Table \ref{tab:complex} shows that Fast-ABS is better than ES in terms of switching latency because a small number of received beams are used to extract three tuples in Fast-ABS ($M_{\rm ABS}\ll M_{\rm ES}$). ES causes the long transmission interruption to find the best receiver beam and results in the disappearance of the line in Figure \ref{fig:Imple_result}.
Meanwhile, SSB is periodically broadcasted via the BS by using different transmit beams, obtaining further information regarding the beam quality of the other transmit beams. After sufficient information is achieved, the UE informs the BS to transmit data to another beam and switches a proper receiver beam corresponding to the transmit beam. Therefore, performance improves at stage III. 
Notably, ES takes a long time ($S\times M_{\rm ES}$ slots) to collect such information through the joint transmission and receiver beam sweeping, although data transmission is still maintained during this period.
In summary, Table \ref{tab:complex} shows that in stages II and III, Fast-ABS is faster than ES because of the fact that Fast-ABS only requires a small number of received beams ($M_{\rm ABS}\ll M_{\rm ES}$) to extract three tuples.
At stage IV, module 2 is blocked by a hand. Fast-ABS rapidly switches to module 1 rapidly while ES spends a long time ($S\times M_{\rm ES}\times P$ slots) to change the array and beam pair. If the ES scheme is employed, data cannot be received (the disappearance of the blue line) because the UE switches to module 1 and searches for receiver beams. Finally, at stage V, hand blockage is removed from module 2. As expected, Fast-ABS rapidly switches to module 2 while ES spends a long time ($S\times M_{\rm ES}\times P$ slots) to detect the antenna module and beam pair.
In summary, regardless of the presence of hand blockage (stages IV and V), Fast-ABS, unlike ES, does not need to perform beam scanning again; instead, it switches to a proper beam directly on the basis of the virtual channel. Therefore, as shown in Table \ref{tab:complex}, Fast-ABS does not entail time complexity in the two stages.

\begin{figure*}
    \begin{center}
        \resizebox{6.75in}{!}{%
            \includegraphics*{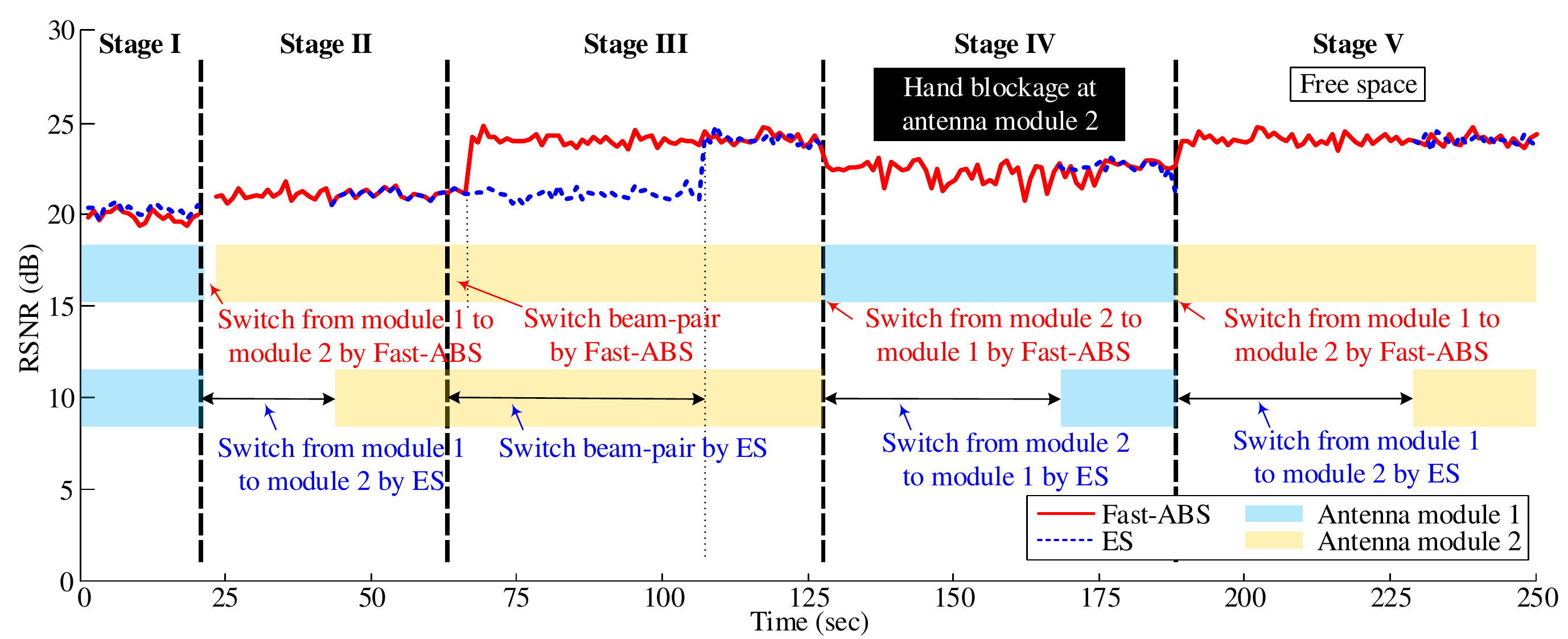} }%
        \caption{RSNR trajectory during Fast-ABS over the mmWave testbed. The corresponding RSNRs of Fast-ABS and ES are plotted with red and blue
        dotted lines, respectively.
            }\label{fig:Imple_result}
    \end{center}
\end{figure*}

\begin{table*}[]
\centering \footnotesize
\caption{Results of subarrays over the mmWave testing platform.}\label{tab:Imple_result}
\begin{tabular}{lrllll}
\toprule
\multicolumn{2}{c}{\multirow{2}{*}{}} & \multicolumn{2}{c}{Antenna module 1} & \multicolumn{2}{c}{Antenna module 2} \\ \cmidrule{3-6}
\multicolumn{2}{c}{} & BS beam 1 & BS beam 2 & BS beam 1 & BS beam 2 \\ \midrule
\multirow{3}{*}{Stage I: Established by initial beam pair} & AoA & \colorbox{light-gray}{130.25$^\circ$} & - & - & - \\
 & $\phi$ & \colorbox{light-gray}{135$^\circ$} & - & - & - \\
 & RSNR & \colorbox{light-gray}{19.86 dB}& - & - & - \\ \midrule
\multirow{3}{*}{Stage II: Change antenna module} & AoA & 130.25$^\circ$ & - & \colorbox{light-gray}{49.80$^\circ$} & - \\
 & $\phi$ & 135$^\circ$ & - & \colorbox{light-gray}{45$^\circ$} & - \\
 & RSNR & 19.86 dB& - & \colorbox{light-gray}{21.25 dB} & - \\ \midrule
\multirow{3}{*}{Stage III: Change beam-pair} & AoA & 130.25$^\circ$ & 116.85$^\circ$ & 49.80$^\circ$ & \colorbox{light-gray}{40.01$^\circ$} \\
 & $\phi$ & 135$^\circ$ & 120$^\circ$ & 45$^\circ$ & \colorbox{light-gray}{45$^\circ$} \\
 & RSNR & 19.86 dB & 21.95 dB & 21.25 dB & \colorbox{light-gray}{24.69 dB} \\ \midrule
\multirow{3}{*}{Stage IV: Hand blockage at module 2} & AoA & 128.13$^\circ$ & \colorbox{light-gray}{115.90$^\circ$} & - & - \\
 & $\phi$ & 135$^\circ$ & \colorbox{light-gray}{120$^\circ$} & - & - \\
 & RSNR & 21.66 dB & \colorbox{light-gray}{22.74 dB} & 13.44 dB & 13.43 dB \\ \midrule
\multirow{3}{*}{Stage V: Hand away from module 2} & AoA & 130.25$^\circ$ & 116.85$^\circ$ & 49.80$^\circ$ & \colorbox{light-gray}{40.01$^\circ$} \\
 & $\phi$ & 135$^\circ$ & 120$^\circ$ & 45$^\circ$ & \colorbox{light-gray}{45$^\circ$} \\
 & RSNR & 19.86 dB & 21.95 dB & 21.25 dB & \colorbox{light-gray}{24.69 dB} \\
\bottomrule
\end{tabular}
\end{table*}

\begin{table}[]
\centering 
\caption{Comparison of ES and Fast-ABS complexity.}\label{tab:complex}
\begin{tabular}{lcc}
\hline
 & ES & Fast-ABS \\ \hline
Stage I: Established by initial beam pair & - & - \\ \hline
Stage II: Change antenna module & $M_{\rm ES}$ & $M_{\rm ABS}$ \\ \hline
Stage III: Change beam-pair & $S\times M_{\rm ES}$ & $S\times M_{\rm ABS}$ \\ \hline
Stage IV: Hand blockage at module 2 & $S\times M_{\rm ES}\times P$ & 0 \\ \hline
Stage V: Hand away from module 2 & $S\times M_{\rm ES}\times P$ & 0 \\ \hline
&&unit: slots
\end{tabular}
\end{table}

\section{Conclusion}

We addressed the problem of determining the optimal antenna module and beam pair for a mobile phone consisting of multiple antenna modules with only one antenna module that could be powered on at a time. We presented Fast-ABS by exploiting the propagation invariant property. In the proposed method, only one antenna module was used for the reception to predict the best beam of other antenna modules. In particular, Fast-ABS estimated the corresponding channel of other antenna modules by extracting three tuples of each multipath associated with one antenna module. We performed a thorough theoretical analysis to demonstrate that if a proper set of beams could be used, then the low MSE of the three-tuple parameters could be achieved with a small number of received beam observations. Simulations demonstrated that Fast-ABS might achieve almost the same performance as that of a method through which all angles and antenna modules were scanned. Furthermore, we implemented and evaluated Fast-ABS in a 5G NR device. A series of experimental results supported our arguments and indicated that the performance of Fast-ABS was close to that of an oracle solution even in complex NLoS scenarios.

\section*{Appendix A}
Substituting \eqref{eq:dev_F} into \eqref{eq:fisher_ele_sim}, we can derive $\qF(\qpsi) = \frac{2}{\sigma^2_{z}}\qP \odot \qQ$, 
where
\begin{equation}
\qP =
\begin{bmatrix}
N_s &|g|N_s  & |g|N_s &\pi|g|(N_s-1) \\
\cdot &|g|^2N_s & |g|^2N_s & \pi|g|^2(N_s-1)\\
\cdot & \cdot &|g|^2N_s  & \pi|g|^2(N_s-1)\\
\cdot & \cdot & \cdot &\frac{2\pi^2|g|^2(2N_s^2-3N_s+1)}{3N_s}
\end{bmatrix},
\end{equation}
is a symmetric matrix,
\begin{equation*}
\qQ =
\begin{bmatrix}
\Re\{A\} & -\Im\{A\} & \Re\{B\} & \Im\{A\}\\\
\Im\{A\} & \Re\{A\} & \Im\{B\} & -\Re\{A\}\\
\Re\{C\} & -\Im\{C\} & \Re\{D\} & \Im\{C\}\\
-\Im\{A\} & -\Re\{A\} & -\Im\{B\} & \Re\{A\}
\end{bmatrix},
\end{equation*}
and
\begin{subequations} \label{eq:AtoD}
\begin{align}
A &\triangleq \sum_{m=1}^{M_s}A_{m}^*(\theta)A_{m}(\theta),
&&B \triangleq \sum_{m=1}^{M_s}A_{m}^*(\theta){A}_{m}'(\theta),\\
C &\triangleq \sum_{m=1}^{M_s}{A}_{m}'^*(\theta)A_{m}(\theta),
&&D \triangleq \sum_{m=1}^{M_s}{A}_{m}'^*(\theta){A}_{m}'(\theta).
\end{align}
\end{subequations}
The CRLB of each parameter can be obtained from \eqref{eq:CRLB} as
\begin{subequations}
\begin{align}
\sigma^2_{|g|} &\triangleq [(\qF^{-1}(\qpsi))]_{1, 1} = \frac{\sigma^2_{z}(\Im\{B\}\Im\{C\}+\Re\{A\}\Re\{D\})}{2N_sJ}, \label{eq:sig_absg} \\
\sigma^2_{\angle g} &\triangleq [(\qF^{-1}(\qpsi))]_{2, 2} = -\frac{\sigma^2_{z}(3J(1-N_s)-Q(1+N_s))}{2N_s|g|^2\Re\{A\}(N_s+1)J},\\
\sigma^2_{\theta} &\triangleq [(\qF^{-1}(\qpsi))]_{3, 3} = \frac{\sigma^2_{z}|A|^2}{2N_s|g|^2J}, \label{eq:sig_theta}\\
\sigma^2_{\tau} &\triangleq [(\qF^{-1}(\qpsi))]_{4, 4} = \frac{3N_s\sigma^2_{z}}{2\pi^2|g|^2\Re\{A\}(N_s^2-1)},
\end{align}
\end{subequations}
where
\begin{align}
J &\triangleq |A|^2\Re\{D\}-\Im\{A\}\Im\{B\}\Re\{C\}-\Im\{A\}\Im\{C\}\Re\{B\} \notag \\
&~~~~+\Im\{B\}\Im\{C\}\Re\{A\}-\Re\{A\}\Re\{B\}\Re\{C\}, \label{eq:J}\\
Q &\triangleq \Re\{A\}^2\Re\{D\}-\Re\{A\}\Re\{B\}\Re\{C\}.
\end{align}

{\renewcommand{\baselinestretch}{1.1}
\begin{footnotesize}

\end{footnotesize}}

\end{document}